\documentclass[12pt]{article}
\usepackage[margin=1in]{geometry}
\usepackage{setspace}
\usepackage[margin=1.5cm, labelfont=bf]{caption}
\usepackage{subcaption}
\usepackage{graphicx}
\usepackage{fancyhdr}
\usepackage{xcolor}

\usepackage{titlesec}
\titleformat*{\section}{\Large\scshape}
\titleformat*{\subsection}{\large\scshape}
\titleformat*{\subsubsection}{\normalsize\itshape}

\usepackage{amssymb,amsmath,amsthm,amsfonts,mathtools,array}
\usepackage{verbatim, tikz, bbm}
\usepackage{xfrac} 
\usepackage[mathscr]{euscript} 
\usepackage{empheq} 
\usepackage{enumitem}
\usepackage{subeqnarray}

\theoremstyle{plain} \numberwithin{equation}{section}
\newtheorem{theorem}{Theorem}[section]

\usepackage{url}
\usepackage{hyperref}
\definecolor{linkblue}{HTML}{1EB5EB}
\definecolor{linkgreen}{HTML}{06BA63}
\definecolor{linkred}{HTML}{E3170A}

\hypersetup{
	pdftitle={Ion-Acoustic Wave Dynamics in a Two-Fluid Plasma},
	pdfauthor={Emily Kelting},
	pdfsubject={partial differential equations},
	pdfkeywords={Plasma, Euler-Poisson, Ion-Acoustic Waves, solitons},
	colorlinks=true,
	linkcolor=linkblue,
	citecolor=linkgreen,
	urlcolor=linkred
}
\usepackage{authblk}
\usepackage[nameinlink]{cleveref}
\usepackage{bigints}
\usepackage{rotating}

\catcode`,\active

\catcode`\,12

\usepackage{listings}

\definecolor{codegreen}{HTML}{06BA63}
\definecolor{codegray}{rgb}{0.5,0.5,0.5}
\definecolor{codepurple}{rgb}{0.58,0,0.82}
\definecolor{codeblue}{HTML}{43C1EF}
\definecolor{backcolour}{HTML}{F9F6F0}

\lstdefinestyle{mystyle}{
    backgroundcolor=\color{backcolour},   
    commentstyle=\color{codegreen},
    keywordstyle=\color{codeblue},
    numberstyle=\tiny\color{codegray},
    stringstyle=\color{codepurple},
    basicstyle=\ttfamily\footnotesize,
    breakatwhitespace=false,         
    breaklines=true,                 
    captionpos=b,                    
    keepspaces=true,                 
    numbers=left,                    
    numbersep=5pt,                  
    showspaces=false,                
    showstringspaces=false,
    showtabs=false,                  
    tabsize=2
}

\usepackage{twemojis}

\begin{document} 
\allowdisplaybreaks
\thispagestyle{empty}
\setstretch{1.15}

\begin{center}
    \textsc{\Large{Ion-Acoustic Wave Dynamics in a Two-Fluid Plasma}}
\\[20pt]
    {\large{Emily K. Kelting\textsuperscript{\small{\twemoji{cloud with lightning}}} and J. Douglas Wright\textsuperscript{\footnotesize{\twemoji{dragon}}}}}
    \\[10pt]
\end{center}
\begin{minipage}[c]{0.45\textwidth}
\centering\textsuperscript{\small{\twemoji{cloud with lightning}}}\textit{Department of Mathematics \\ University of New England \\ Biddeford, ME 04005}
\end{minipage}
\hfill
\begin{minipage}[c]{0.45\textwidth}
\centering\textsuperscript{\footnotesize{\twemoji{dragon}}}\textit{Department of Mathematics \\ Drexel University \\ Philadelphia, PA 19104}
\end{minipage}
\\[10pt]

\hrule
\begin{abstract}
    Plasma is a medium containing free electrons and cations, where each particle group behaves as a conducting fluid with a single velocity and temperature in the presence of electromagnetic fields. The difference in roles electrons and ions play define the two-fluid description of plasma. This paper examines ion-acoustic waves generated by the particles in both hot and cold plasma using a collisionless  ``Euler-Poisson'' (EP) system. Employing phase-space asymptotic analysis, we establish that for specific wave speeds, EP acquires homoclinic orbits at the steady-state equilibrium and consequently, traveling waves. Combining \texttt{python} and \texttt{Wolfram Mathematica}, we captured visualizations of such behavior in one spatial dimension.
\end{abstract}

\\[6pt]
\hrule

\pagestyle{fancy}
\section{Introduction}

Plasma, a state of matter consisting of free electrons and cations, exhibits complex behaviors under the influence of electromagnetic fields. The two-fluid Euler-Poisson (EP) system is a powerful framework for studying these behaviors by treating the particle groups as separate conducting fluids with distinct velocities and temperatures. Applying this system to hot and cold plasma scenarios, we provide a detailed analysis of the interactions and dynamics within the medium -- specifically the propagation of ion-acoustic waves. 

The collisionless Euler-Poisson (EP) system for a two-fluid plasma is 
\begin{equation}
    \label{eq::5hotrescale}
    \begin{aligned}
        \partial_t n_+ + \partial_x(n_+ v_+) = 0, 
         \\[4pt]
        \partial_t n_- + \partial_x(n_- v_-) = 0,  \\[4pt]
        \partial_t v_+ + \tfrac{1}{2}\partial_x v_+^2 + \tau_i \partial_x \ln(n_+) + \partial_x\phi = 0,
        \\[4pt]
        m_e \left(\partial_t v_- + \tfrac{1}{2}\partial_x v_-^2\right) + \partial_x \ln(n_-) - \partial_x\phi = 0,
         \\[4pt]
        \partial_{xx}\phi - n_- + n_+ = 0.
    \end{aligned}
\end{equation}
Here, \(n_{\pm}(x,t)\) denotes the particle number density at spatial point \(x\in\mathbb{R}\) and time \(t\geq 0\), \(v_{\pm}(x,t)\) represents the velocity, and \(\phi(x,t)\) is the electric potential. The subscripts represent the particle groups by their charge with ``+'' being the cations and ``--'' the electrons. The parameter \(m_e\) is the mass of the electron, and \(\tau_i\) is the temperature of the ion. The model is scaled so the ion mass is at unity, making the corresponding electron mass \(m_e \approx 0.00054551\). Additionally, by setting the electron temperature to unity, the plasma can be identified as ``hot'' when \(\tau_i=1\)  (thermodynamic equilibrium) and ``cold'' when \(\tau_i=0\) \cite{krishan}.

\subsection{Motivation}

Prior studies on plasma show the cold single-fluid description possesses ion-acoustic solitary wave solutions \cite{guo,haragus}. Haragus and Scheel established such waves are spectrally stable at low amplitudes in one spatial dimension \cite{haragus}. However, not much is known about their existence in the two-fluid system. So far, Kelting and Wright have established that Korteweg–de Vries (KdV) equations govern the dynamics of the two-fluid EP system in the long-wavelength limit for both hot and cold plasmas. An incredibly intriguing aspect of the KdV equation is its ability to support solitons. Observing these facts, it is important to check if \eqref{eq::5hotrescale} has traveling waves as well.
\section{Have wave -- Will travel}\label{trav_wave}

\subsection{An ansatz}

We begin by making a traveling wave ansatz for waves of speed \(\mu\),
\begin{align*}
    n_{\pm}(x,t) & = 1 + \tilde{n}_{\pm}(x-\mu t) \coloneqq 1+ \tilde{n}_{\pm}(\xi), \\
    v_{\pm}(x,t) & = \tilde{v}_{\pm}(x-\mu t) \coloneqq \tilde{v}_{\pm}(\xi), \\
    \phi(x,t) & = \tilde{\phi}(x-\mu t) \coloneqq \tilde{\phi}(\xi).
\end{align*}
Substituting into \eqref{eq::5hotrescale}, we have a system depending on only one spatial variable's derivative,
\vspace{-12pt}
\begin{subequations}
    \label{eq:5hot_TravWave}
    \begin{align}
        \label{eq:5travWave_dens}
        -\mu \tilde{n}_\pm' + \tilde{v}_\pm' + (\tilde{n}_\pm \tilde{v}_\pm)' & = 0, \\
        \label{eq:5travWave_vel+}
        -\mu \tilde{v}_+' + \tfrac{1}{2}(\tilde{v}_+^2)' + \tau_i (\ln(1+\tilde{n}_+))' + \tilde{\phi}' & = 0, \\
        \label{eq:5travWave_vel-}
        -m_e\mu \tilde{v}_-' + \tfrac{m_e}{2}(\tilde{v}_-^2)' + (\ln(1+\tilde{n}_-))' - \tilde{\phi}' & = 0, \\
        \label{eq:5travWave_phi}
        \tilde{\phi}'' - \tilde{n}_- + \tilde{n}_+ & = 0.
    \end{align}
\end{subequations}

The first four equations can be integrated with respect to \(\xi\). Because traveling waves are asymptotically null, i.e. \(\lim_{\xi\to\pm\infty} \big( \tilde{n}_{\pm}(\xi), \tilde{v}_{\pm}(\xi), \tilde{\phi}(\xi) \big) = 0\), the constants of integration are zero. From \eqref{eq:5travWave_dens}, we determine a relationship between the particle density and velocity waves, 
\begin{equation}\label{eq:travwave_dens_vel}
\tilde{n}_\pm = \frac{\tilde{v}_\pm}{\mu-\tilde{v}_\pm}.
\end{equation}
Next, we integrate \eqref{eq:5travWave_vel+} to get
\begin{equation*}
    \tilde{\phi} = \mu \tilde{v}_+ - \tfrac{1}{2}(\tilde{v}_+^2) - \tau_i \ln(1+\tilde{n}_+).
\end{equation*}
Replacing \(\tilde{n}_+\) with \eqref{eq:travwave_dens_vel}, the relationship is then only between the ion velocity wave and the electric potential wave. Hence,
\begin{equation}\label{eq:travwave_phi_vel+}
    \tilde{\phi} = \mu \tilde{v}_+ - \tfrac{1}{2}(\tilde{v}_+^2) + \tau_i \ln\left(\mu-\tilde{v}_+ \right)  - \tau_i \ln\mu.
\end{equation}
Similarly, \eqref{eq:5travWave_vel-} yields
\begin{equation}\label{eq:travwave_phi_vel-}
    \tilde{\phi} = -m_e\mu \tilde{v}_- + \tfrac{m_e}{2}(\tilde{v}_-^2) - \ln\left(\mu - \tilde{v}_- \right) + \ln\mu.
\end{equation}

Right now, these calculations have established \eqref{eq:5travWave_phi} can be rewritten as 
\begin{equation*}
    \tilde{\phi}'' = \frac{\tilde{v}_-}{\mu-\tilde{v}_-} - \frac{\tilde{v}_+}{\mu-\tilde{v}_+}.
\end{equation*}
However, if we define new functions \(k_\pm(\tilde{\phi}) \coloneqq \tilde{v}_\pm\) such that their inverses satisfy \eqref{eq:travwave_phi_vel+} and \eqref{eq:travwave_phi_vel-}, we can rework it as
\begin{equation}\label{eq:phiODE}
    \tilde{\phi}'' = \frac{k_-(\tilde{\phi})}{\mu-k_-(\tilde{\phi})} - \frac{k_+(\tilde{\phi})}{\mu-k_+(\tilde{\phi})}.
\end{equation}
This is a single differential equation of one variable. Investigation into the behavior of \eqref{eq:phiODE} in its phase space will provide information on the existence of solitary wave solutions to five-equation partial differential system \eqref{eq::5hotrescale}. 

To begin our analysis, a relationship between \(\tilde{\phi}'(\xi)\) and \(\tilde{\phi}(\xi)\) is necessary. Multiplying both sides of \eqref{eq:phiODE} by \(\tilde{\phi}'(\xi)\) and integrating provides a useful left-hand-side of the equation, but not so much on the right: 
\begin{equation*}
    \tfrac{1}{2}(\tilde{\phi}')^2 = \int \frac{k_-(\tilde{\phi})\tilde{\phi}'}{\mu-k_-(\tilde{\phi})} d\xi - \int \frac{k_+(\tilde{\phi})\tilde{\phi}'}{\mu-k_+(\tilde{\phi})} d\xi
\end{equation*}
Though there are no explicit formulas for \(k_{\pm}(\tilde{\phi}(\xi))\), we can utilize the chain rule and substitution to find a nearly explicit conserved quantity. Solving for right-hand-side quantities,
\begin{align*}
    \int \frac{k_-(\tilde{\phi})\tilde{\phi}'}{\mu-k_-(\tilde{\phi})}d\xi & = \int \frac{\tilde{v}_-\left(-m_e\mu + m_e\tilde{v}_- + \frac{1}{\mu - \tilde{v}_-}\right)\tilde{v}_-'}{\mu-\tilde{v}_-}d\xi \\
    \int \frac{k_-(\tilde{\phi})\tilde{\phi}'}{\mu-k_-(\tilde{\phi})}d\xi & = \ln\lvert \mu - \tilde{v}_-\rvert - \tfrac{m_e}{2}(\tilde{v}_-)^2 + \frac{\mu}{\mu - \tilde{v}_-} + const \\
    & = \ln\lvert \mu - k_-(\tilde{\phi})\rvert - \tfrac{m_e}{2}(k_-(\tilde{\phi}))^2 + \frac{\mu}{ \mu - k_-(\tilde{\phi})} + const,
\end{align*}
and
\begin{align*}
    \int \frac{k_+(\tilde{\phi})\tilde{\phi}'}{\mu-k_+(\tilde{\phi})}d\xi & = \int \frac{\tilde{v}_+\bigg(\mu - \tilde{v}_+ - \frac{\tau_i}{\mu - \tilde{v}_+}\bigg)\tilde{v}_+'}{\mu-\tilde{v}_+}d\xi \\
    & = -\tau_i \ln\lvert \mu - \tilde{v}_+\rvert + \tfrac{1}{2}(\tilde{v}_+)^2 - \frac{\mu \tau_i}{\mu - \tilde{v}_+} + const \\
    & = -\tau_i \ln\lvert \mu - k_+(\tilde{\phi})\rvert + \tfrac{1}{2}(k_+(\tilde{\phi}))^2 - \frac{\mu \tau_i}{\mu - k_+(\tilde{\phi})} + const. 
\end{align*}
To find the constants of integration, recall traveling waves are asymptotically null. That is, each integral needs to approach zero as \(\xi\to 0\). Therefore,
\begin{align*}
    \int \frac{k_-(\tilde{\phi})\tilde{\phi}'}{\mu-k_-(\tilde{\phi})}d\xi & = \ln\lvert \mu - k_-(\tilde{\phi})\rvert - \tfrac{m_e}{2}(k_-(\tilde{\phi}))^2 + \frac{\mu}{ \mu - k_-(\tilde{\phi})} - ( \ln\mu + 1), \\
    \int \frac{k_+(\tilde{\phi})\tilde{\phi}'}{\mu-k_+(\tilde{\phi})}d\xi & = -\tau_i \ln\lvert \mu - k_+(\tilde{\phi})\rvert + \tfrac{1}{2}(k_+(\tilde{\phi}))^2 - \frac{\mu \tau_i}{\mu - k_+(\tilde{\phi})} + \tau_i (\ln\mu + 1).
\end{align*}
Next, we note \(\mu > k_\pm(\tilde{\phi})\) because of the domains for \eqref{eq:travwave_phi_vel+} and \eqref{eq:travwave_phi_vel-}. Hence, 
\begin{equation*}
    \ln\lvert \mu - k_\pm(\tilde{\phi}) \rvert = \ln(\mu - k_\pm(\tilde{\phi})).
\end{equation*}

\subsection{Unleashing the Potential Energy}

We will use a potential energy argument to study the behavior of \(\tilde{\phi}\) in phase space. Let \(\mathcal{V}(\tilde{\phi})\) represent the ``potential energy'' curve which shows how energy is distributed among the electrons and cations as the wave propagates through plasma. Here,
\begin{equation*}
    \mathcal{V}(\tilde{\phi}) = - \int \frac{k_-(\tilde{\phi})\tilde{\phi}'}{\mu-k_-(\tilde{\phi})}d\xi + \int \frac{k_+(\tilde{\phi})\tilde{\phi}'}{\mu-k_+(\tilde{\phi})} d\xi
\end{equation*}
so that \(\tfrac{1}{2}(\tilde{\phi}')^2 = -\mathcal{V}(\tilde{\phi}).\) Utilizing our previous calculations, we can rewrite the potential energy as
\begin{equation}\label{eq:PotEnergy}
    \begin{aligned}
     \mathcal{V}(\tilde{\phi}) = & -\ln( \mu - k_-(\tilde{\phi})) + \tfrac{m_e}{2}(k_-(\tilde{\phi}))^2 - \frac{\mu}{\mu - k_-(\tilde{\phi})} + ( \ln\mu + 1) \\
     & - \tau_i \ln (\mu - k_+(\tilde{\phi}) ) + \tfrac{1}{2}(k_+(\tilde{\phi}))^2 - \frac{\mu \tau_i}{\mu - k_+(\tilde{\phi})} + \tau_i (\ln\mu + 1)
     \end{aligned}
\end{equation}

The domain of such a function is 
\begin{equation*}
    \left[k_-^{-1}(\mu - \tfrac{1}{\sqrt{m_e}}), \, k_+^{-1}(\mu - \sqrt{\tau_i})\right].
\end{equation*}
To obtain this, we look at the ranges of \(k_\pm^{-1}(\tilde{v}_\pm)\) defined by \eqref{eq:travwave_phi_vel+} and \eqref{eq:travwave_phi_vel-}. For instance, consider \(k_+^{-1}(\tilde{v}_+)\). We know for all \(\tilde{v}_+\),
\begin{equation*}
    k_+^{-1}(\tilde{v}_+)'' = -1 - \frac{\tau_i}{(\mu - \tilde{v}_+)^2} < 0,
\end{equation*}
so if \(k_+^{-1}(\tilde{v}_+)\) has a critical point, it must be a maximum. Following this path, \(k_+^{-1}(\tilde{v}_+)' = 0\) when \((\mu - \tilde{v}_+)^2 - \tau_i = 0\). The only solution in the domain is~\(\tilde{v}_+~=~\mu~-~\sqrt{\tau_i}\). Thus, the range of \(k_+^{-1}(\tilde{v}_+)\), and moreover domain of \(k_+(\tilde{\phi})\), is \(\left(-\infty,k_+^{-1}(\mu - \sqrt{\tau_i})\right]\). Through a similar process, we can find the range of \(k_-^{-1}(\tilde{v}_-)\), or the domain of \(k_-(\tilde{\phi})\), to be \(\left[k_-^{-1}(\mu - \tfrac{1}{\sqrt{m_e}}),\infty\right)\).


\subsection{Looking for Homoclinic Orbits}

We want to turn our traveling wave ansatz into a certainty. Therefore, we need to look for homoclinic orbits -- a path in the phase space that starts and ends at the same equilibrium point -- because they translate to traveling waves in the single spatial dimension. The potential energy argument states these orbits occur when the following conditions are all satisfied:
\begin{equation*}
    \mathcal{V}(0) = 0, \quad \mathcal{V}'(0)=0, \quad \mathcal{V}''(0)<0, \quad \exists \tilde{\phi}^* > 0 \;:\; \mathcal{V}(\tilde{\phi}^*) > 0. 
\end{equation*}
We found that \eqref{eq:PotEnergy} meets these requirements at specific wave speeds with \(\tilde{\phi}^* = k_+^{-1}(\mu - \sqrt{\tau_i})\), the right-most domain point. The results are summarized in Theorem~\ref{thm:PE_homoclinic}.
\begin{center}{\rule{4cm}{0.4pt}}\end{center}

\begin{theorem}\label{thm:PE_homoclinic}
    The ``potential energy'' system for a two-fluid plasma,
    \begin{align*}
        \frac{(\tilde{\phi}')^2}{2} & = -\mathcal{V}(\tilde{\phi}), \\[5pt]
        \mathcal{V}(\tilde{\phi}) & = -\ln( \mu - k_-(\tilde{\phi})) + \tfrac{m_e}{2}(k_-(\tilde{\phi}))^2 - \frac{\mu}{\mu - k_-(\tilde{\phi})} + ( \ln\mu + 1) \\
        & \quad\; - \tau_i \ln (\mu - k_+(\tilde{\phi}) ) + \tfrac{1}{2}(k_+(\tilde{\phi}))^2 - \frac{\mu \tau_i}{\mu - k_+(\tilde{\phi})} + \tau_i (\ln\mu + 1), \\[5pt]
        \tilde{\phi} & \in \left[k_-^{-1}(\mu - \tfrac{1}{\sqrt{m_e}}), \, k_+^{-1}(\mu - \sqrt{\tau_i})\right] 
    \end{align*}
    acquires homoclinic orbits for wave speeds \(\mu\) such that
    \begin{align*}
         \frac{1}{\sqrt{m_e}} & > \mu > \sqrt{\tau_i}, \\
        \mu & > \sqrt{\frac{1+\tau_i}{1+m_e}}, \\
        \frac{m_e \mu^2}{2}\left(1 - e^{-\tfrac{1}{2} \left(\mu ^2 - \tau_i\right) - \tau_i \ln \left(\tfrac{\sqrt{\tau_i}}{\mu }\right)}\right)^2 & - e^{\tfrac{1}{2} \left(\mu ^2 - \tau_i\right) + \tau_i \ln \left(\frac{\sqrt{\tau_i}}{\mu }\right)} + (\mu - \sqrt{\tau_i})^2 + 1 > 0.
    \end{align*}
\end{theorem}

\begin{proof}
For homoclinic orbits, we need wave speeds \(\mu\) such that 
\begin{equation*}
    \mathcal{V}(0) = 0, \quad \mathcal{V}'(0)=0, \quad \mathcal{V}''(0)<0, \quad \mathcal{V}(k_+^{-1}(\mu - \sqrt{\tau_i})) > 0. 
\end{equation*}
We will walk through each requirement to find the conditions under which they are met. Before starting, it is important to note that \(k_-^{-1}(0) = 0\) gives \(k_-(0)=0\). However, there is a catch in the domain of \(\mathcal{V}(\tilde{\phi})\) to get \(k_+(0) = 0\). Since \(\tilde{\phi} \in \left[k_-^{-1}(\mu - \tfrac{1}{\sqrt{m_e}}), \, k_+^{-1}(\mu - \sqrt{\tau_i})\right]\),
\begin{equation*}
    \tilde{v}_\pm = k_\pm(\tilde{\phi})\in \left[\mu - \tfrac{1}{\sqrt{m_e}}, \, \mu - \sqrt{\tau_i}\right].
\end{equation*}
Because traveling waves are asymptotically null, we need \(\tilde{v}_+ = 0\) to exist. So, we require \(\mu \geq \sqrt{\tau_i}\). 

\vspace{11pt}

\noindent\underline{\(\mathcal{V}(0) = 0\)}: By definition.

\vspace{11pt}

\noindent\underline{\(\mathcal{V}'(0) = 0\)}:
\begin{align*}
    \mathcal{V}'(0) & = \frac{k_-'(0)}{\mu - k_-(0)} + m_e k_-(0)k_-'(0) - \frac{\mu k_-'(0)}{ (\mu - k_-(0))^2} \\
     & \quad\; +\frac{\tau_i k_+'(0)}{\mu - k_+(0)} + k_+(0)k_+'(0) - \frac{\mu \tau_i k_+'(0)}{(\mu - k_+(0))^2} \\[5pt]
     & = \frac{k_-'(0)}{\mu} - \frac{k_-'(0)}{ \mu} + \frac{\tau_i k_+'(0)}{\mu} - \frac{\tau_i k_+'(0)}{\mu} \\
     & = 0.
\end{align*}

\vspace{11pt}

\noindent\underline{\(\mathcal{V}''(0) < 0\)}:
%
%
To start,
\begin{align*}
    \mathcal{V}''(0) & = \frac{(\mu - k_-(0))k_-''(0) + (k_-'(0))^2}{(\mu - k_-(0))^2} + m_e k_-(0) k_-''(0) + m_e(k_-'(0))^2 \\
    & \quad\; - \frac{\mu(k_-(0) - \mu)^2 k_-''(0) - 2\mu (k_-'(0))^2 (k_-(0) - \mu) }{(k_-(0) - \mu)^4} \\
    & \quad\; + \frac{\tau_i(\mu - k_+(0))k_+''(0) + \tau_i(k_+'(0))^2}{(\mu - k_+(0))^2} + k_+(0) k_+''(0) + (k_+'(0))^2 \\
    & \quad\; - \frac{\mu\tau_i(k_+(0) - \mu)^2 k_+''(0) - 2\mu\tau_i (k_+'(0))^2 (k_+(0) - \mu) }{(k_+(0) - \mu)^4}.
\end{align*}
We can rewrite this as 
\begin{equation*}
    \mathcal{V}''(0) = (k_-'(0))^2 (m_e - \mu^{-2}) + (k_+'(0))^2( 1 - \tau_i \mu^{-2}).
\end{equation*}
To continue our calculation, we need information about \(k_\pm'(\tilde{\phi})\). Applying Chain Rule to \eqref{eq:travwave_phi_vel+}, we know
\begin{align*}
    \frac{d}{d\tilde{\phi}}(\tilde{\phi}) & = \frac{d}{d\tilde{\phi}}\left(\mu \tilde{v}_+ - \tfrac{1}{2}(\tilde{v}_+^2) + \tau_i \ln\left(\mu-\tilde{v}_+ \right)  - \tau_i \ln\mu \right) \\[5pt]
    \implies 1 & = \left(\mu - \tilde{v}_+ - \frac{\tau_i}{\mu - \tilde{v}_+}\right)\frac{d\tilde{v}_+}{d\tilde{\phi}}.
\end{align*}
Thus, 
\begin{equation*}
    k_+'(\tilde{\phi}) = \frac{\mu - k_+(\tilde{\phi})}{(\mu - k_+(\tilde{\phi}))^2 - \tau_i} .
\end{equation*}
By \eqref{eq:travwave_phi_vel-}, we also have
\begin{align*}
    \frac{d}{d\tilde{\phi}}(\tilde{\phi}) & = \frac{d}{d\tilde{\phi}}\left(-m_e\mu \tilde{v}_- + \tfrac{m_e}{2}(\tilde{v}_-^2) - \ln\left(\mu - \tilde{v}_- \right) + \ln\mu\right) \\[5pt]
    \implies 1 & = \left(-m_e \mu + m_e \tilde{v}_- + \frac{1}{\mu - \tilde{v}_-}\right)\frac{d\tilde{v}_-}{d\tilde{\phi}}.
\end{align*}
Hence,
\begin{equation*}
    k_-'(\tilde{\phi}) = \frac{\mu - k_-(\tilde{\phi})}{-m_e(\mu - k_-(\tilde{\phi}))^2 + 1}.
\end{equation*}
This means
\begin{align*}
    \mathcal{V}''(0) & = (k_-'(0))^2 (m_e - \mu^{-2}) + (k_+'(0))^2( 1 - \tau_i \mu^{-2}) \\[5pt]
    & = \left(\frac{\mu - k_-(0)}{-m_e(\mu - k_-(0))^2 + 1}\right)^2 (m_e - \mu^{-2}) + \left(\frac{\mu - k_+(0)}{(\mu - k_+(0))^2 - \tau_i}\right)^2( 1 - \tau_i \mu^{-2}) \\[5pt]
    & = \left(\frac{\mu}{-m_e \mu^2 + 1}\right)^2 (m_e - \mu^{-2}) + \left(\frac{\mu}{\mu^2 - \tau_i}\right)^2( 1 - \tau_i \mu^{-2}) \\[5pt]
    & = \left(\frac{\mu^2}{(-m_e \mu^2 + 1)^2}\right) \left(\frac{m_e\mu^2 - 1}{\mu^2}\right) + \left(\frac{\mu^2}{(\mu^2 - \tau_i)^2}\right)\left(\frac{\mu^2 - \tau_i}{\mu^2}\right) \\[5pt]
    & = \frac{1}{m_e\mu^2 - 1} + \frac{1}{\mu^2 - \tau_i}.
\end{align*}

Because we require \(\mathcal{V}''(0) < 0\),
\begin{equation*}
     \frac{1}{\mu^2 - \tau_i} < \frac{1}{1 - m_e\mu^2}.
\end{equation*}
There are a few cases for the values of \(\mu\) under which this inequality could be true. Recall that we already require \(\mu \geq \sqrt{\tau_i}\). We now need to go further and say \(\mu > \sqrt{\tau_i}\). Alongside this, two possibilities arise. If \(1 - m_e\mu^2 < 0\), then \(\mu^2 > \tfrac{1}{m_e} > \tau_i\), and yields
\begin{align*}
     \frac{1}{\mu^2 - \tau_i} < \frac{1}{1 - m_e\mu^2} & \implies 1 - m_e\mu^2 > \mu^2 - \tau_i \\
     & \implies \frac{1 + \tau_i}{1 + m_e} > \mu^2,
\end{align*}
a contradiction as \(\tfrac{1 + \tau_i}{1 + m_e} \in [\tfrac{1}{1+m_e},\tfrac{2}{1+m_e}] < \tfrac{1}{m_e} \). Therefore, it must be that \(1 - m_e\mu^2 > 0\), i.e. \(\mu < \tfrac{1}{\sqrt{m_e}}\). Ergo,
\begin{equation*}
    \frac{1}{\mu^2 - \tau_i} < \frac{1}{1 - m_e \mu^2} \; \implies \; 1 - m_e \mu^2 < \mu^2 - \tau_i \; \implies \; \frac{1+\tau_i}{1+m_e} < \mu^2.
\end{equation*}
This implies that the traveling wave speed, \(\mu\), must be faster than the sonic speed for a two-fluid plasma defined in \cite{kelting}. In Figure~\ref{fig:mu_bounds}, we provide a diagram of this boundary (blue dashed line) alongside the following constraint. 

\vspace{11pt}

\noindent \underline{\(\mathcal{V}(k_+^{-1}(\mu - \sqrt{\tau_i})) > 0\)}:
First, define \(k_+^{-1}(\mu - \sqrt{\tau_i}) \coloneqq \tilde{\phi}^*\). Then
\begin{align*}
    \mathcal{V}(\tilde{\phi}^*) & = -\ln( \mu - k_-(\tilde{\phi}^*)) + \tfrac{m_e}{2}(k_-(\tilde{\phi}^*))^2 - \frac{\mu}{\mu - k_-(\tilde{\phi}^*)} + ( \ln\mu + 1) \\
        & \quad\; - \tau_i \ln (\mu - k_+(\tilde{\phi}^*) ) + \tfrac{1}{2}(k_+(\tilde{\phi}^*))^2 - \frac{\mu \tau_i}{\mu - k_+(\tilde{\phi}^*)} + \tau_i (\ln\mu + 1) \\[5pt]
        & = -\ln( \mu - k_-(\tilde{\phi}^*)) + \tfrac{m_e}{2}(k_-(\tilde{\phi}^*))^2 - \frac{\mu}{\mu - k_-(\tilde{\phi}^*)} + ( \ln\mu + 1) \\
        & \quad\; - \tfrac{1}{2}\tau_i \ln \tau_i + \tfrac{1}{2}(\mu - \sqrt{\tau_i})^2 - \mu \sqrt{\tau_i} + \tau_i (\ln\mu + 1)
\end{align*}
Unfortunately, \(k_-(\tilde{\phi}^*)\) is not an easy calculation. But, we can make an estimate. Let us define \(k_-(\tilde{\phi}^*)\coloneqq \tilde{v}_-^*\). Then \(\tilde{\phi}^* = k_-^{-1}(\tilde{v}_-^*)\). From \eqref{eq:travwave_phi_vel+} and \eqref{eq:travwave_phi_vel-}, we know this equality is the same as
\begin{equation*}
    \tfrac{1}{2} \left(\mu ^2 - \tau_i\right) + \tau_i \ln \left(\tfrac{\sqrt{\tau_i}}{\mu }\right) = -m_e\mu \tilde{v}_-^* + \tfrac{m_e}{2}(\tilde{v}_-^*)^2 - \ln\left(\mu - \tilde{v}_-^* \right) + \ln\mu.
\end{equation*}
Thus,
\begin{equation*}
    m_e\mu \tilde{v}_-^* - \tfrac{m_e}{2}(\tilde{v}_-^*)^2 + \ln\left(\mu - \tilde{v}_-^* \right) + \tfrac{1}{2} \left(\mu ^2 - \tau_i\right) + \tau_i \ln \left(\tfrac{\sqrt{\tau_i}}{\mu }\right) - \ln\mu = 0.
\end{equation*}
We now make a rough estimate since \(m_e\approx 0\),
\begin{align*}
    0 & = m_e\mu \tilde{v}_-^* - \tfrac{m_e}{2}(\tilde{v}_-^*)^2 + \ln\left(\mu - \tilde{v}_-^* \right) + \tfrac{1}{2} \left(\mu ^2 - \tau_i\right) + \tau_i \ln \left(\tfrac{\sqrt{\tau_i}}{\mu }\right) - \ln\mu \\
    & \approx \ln\left(\mu - \tilde{v}_-^* \right) + \tfrac{1}{2} \left(\mu ^2 - \tau_i\right) + \tau_i \ln \left(\tfrac{\sqrt{\tau_i}}{\mu }\right) - \ln\mu.
\end{align*}
Solving for \(\tilde{v}_-^*\), we find
\(
    \tilde{v}_-^* = \mu - \mu e^{-\tfrac{1}{2} \left(\mu ^2 - \tau_i\right) - \tau_i \ln \left(\tfrac{\sqrt{\tau_i}}{\mu }\right)}.
\)
Hence, \(\mathcal{V}(\tilde{\phi}^*)\) becomes
\begin{align*}
    \mathcal{V}(\tilde{\phi}^*) & = -\ln\left( \mu e^{-\tfrac{1}{2} \left(\mu ^2 - \tau_i\right) - \tau_i \ln \left(\tfrac{\sqrt{\tau_i}}{\mu }\right)}\right) + \frac{m_e}{2}\left(\mu - \mu e^{-\tfrac{1}{2} \left(\mu ^2 - \tau_i\right) - \tau_i \ln \left(\tfrac{\sqrt{\tau_i}}{\mu }\right)}\right)^2 \\
    & \quad\; - \frac{\mu}{\mu e^{-\tfrac{1}{2} \left(\mu ^2 - \tau_i\right) - \tau_i \ln \left(\tfrac{\sqrt{\tau_i}}{\mu }\right)}} - \tfrac{1}{2}\tau_i \ln \tau_i + \tfrac{1}{2}(\mu - \sqrt{\tau_i})^2 - \mu \sqrt{\tau_i} + (1+\tau_i) (\ln\mu + 1).
\end{align*}
Simplifying, we find we need \(\mu\) such that
\begin{equation*}
    \frac{m_e \mu^2}{2}\left(1 - e^{-\tfrac{1}{2} \left(\mu ^2 - \tau_i\right) - \tau_i \ln \left(\tfrac{\sqrt{\tau_i}}{\mu }\right)}\right)^2 - e^{\tfrac{1}{2} \left(\mu ^2 - \tau_i\right) + \tau_i \ln \left(\frac{\sqrt{\tau_i}}{\mu }\right)} + (\mu - \sqrt{\tau_i})^2 + 1 > 0.
\end{equation*}

Taken as a function of the wave speed \(\mu\), we can plot \(\mathcal{V}(\tilde{\phi}^*)\). The graphs for both hot and cold plasmas are given in Figure~\ref{fig:mu_bounds} (black line). Utilizing these visuals, we see that the bounds for \(\mu\) correspond to the roots of \(\mathcal{V}(\tilde{\phi}^*)\). According to Wolfram Mathematica's \texttt{FindRoot} function, for a cold plasma, \(0~<~\mu~<~1.58535\), and for a hot plasma, \(1~<~\mu~<~1.56991\). The lower bound is consistent with the value of \(\sqrt{\tau_i}\).

\begin{figure}[h]
     \centering
     \begin{subfigure}[b]{0.48\textwidth}
         \centering
         \includegraphics[width=\textwidth, trim={0 0 0 0}, clip]{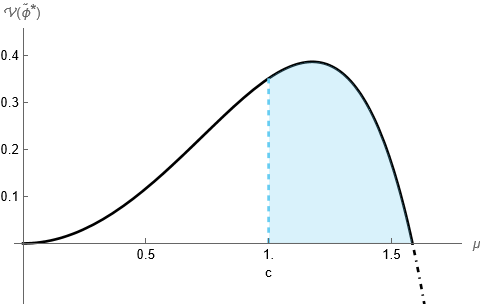}
         \caption{A cold plasma, \(\tau_i=0\)}
         \label{fig:mubound0}
     \end{subfigure}
     \hfill
     \begin{subfigure}[b]{0.48\textwidth}
         \centering
         \includegraphics[width=\textwidth, trim={0 9pt 0 0}, clip]{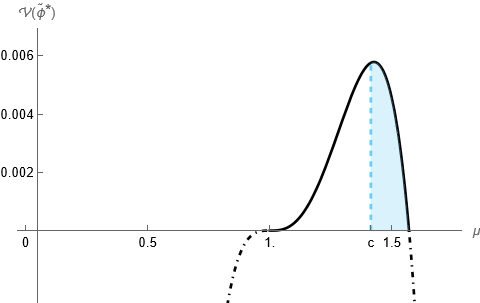}
         \caption{A hot plasma, \(\tau_i=1\)}
         \label{fig:mubound1}
     \end{subfigure}
     \caption{A visual depiction of the \(\mu\) values yielding \(\mathcal{V}(\tilde{\phi}^*)>0\) (black) and \(\mu > c \coloneqq \sqrt{\frac{1+\tau_i}{1+m_e}} \) (blue, dashed) in the two-fluid system.}
        \label{fig:mu_bounds}
\end{figure}
\end{proof}

\subsection{It's just a phase! (space)}

We will now use the \(\mu\) bounds we found in Theorem~\ref{thm:PE_homoclinic} to plot the potential energy function, phase space, and traveling wave in one dimension. To start, we establish an approximation for \(k_\pm(\tilde{\phi})\). We will use a process akin to Newton's Root Finding Method. All we have to do is solve \(k_\pm^{-1}(\tilde{v}_\pm) - \tilde{\phi} = 0\) for \(\tilde{v}_\pm\) as it is equivalent to evaluating \(k_\pm(\tilde{\phi})\). The iterative scheme is 
\begin{equation}\label{eq:newton_scheme}
    \tilde{v}_\pm^{(N+1)} = \tilde{v}_\pm^{(N)} - \frac{k_\pm^{-1}\left(\tilde{v}_\pm^{(N)} \right) - \tilde{\phi}}{d k_\pm^{-1} \left(\tilde{v}_\pm^{(N)} \right)}
\end{equation}
where \(d k_\pm^{-1}(\tilde{v}_\pm )\) is the derivative of \(k_\pm^{-1}(\tilde{v}_\pm )\) with respect to \(\tilde{v}_\pm\). 

For an initial approximation, we define the line between the points \(\left(k_\pm^{-1}(r_0),r_0 \right)\) and \(\left(k_\pm^{-1}(r_1),r_1\right)\),
\begin{equation*}
    \tilde{v}_\pm^{(0)} = r_0 + \frac{r_1 - r_0}{k_\pm^{-1}(r_1) - k_\pm^{-1}(r_0)}\left(\tilde{\phi} - k_\pm^{-1}(r_0)\right)
\end{equation*}
We will take \(r_0 = k_-^{-1}\left(\mu - \tfrac{1}{\sqrt{m_e}}\right)\) and \(r_1 = k_+^{-1}\left(\mu - \sqrt{\tau_i}\right)\) in our calculations. Overall, we find that \(N=5\) gives a sufficient approximation to \(k_\pm(\tilde{\phi})\). As an example with \(\tau_i=1\) and \(\mu=1.5\), we include the approximations for \(k_+(\tilde{\phi})\) at \(N=0\), \(N=1\), and \(N=5\) in Figure~\ref{fig:newton_k+} alongside \(k_+^{-1}(\tilde{v})\) from \eqref{eq:travwave_phi_vel+}.
\begin{figure}[h]
    \centering
    \includegraphics[width = 0.8\textwidth]{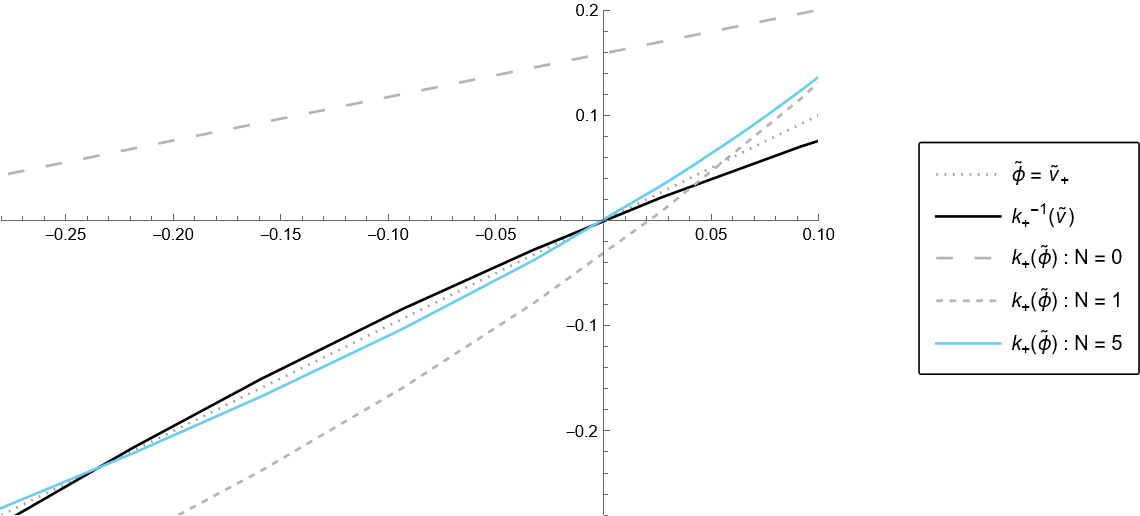}
    \caption{A Newton iterative approximation to \(k_+(\tilde{\phi})\) with \(\mu=1.5\) and \(\tau_i=1\).}
    \label{fig:newton_k+}
\end{figure}

Following our \(k_\pm(\tilde{\phi})\) calculation, we cam establish \(\mathcal{V}(\tilde{\phi})\) and \(\tilde{\phi}'(\xi)\). We then get \(\tilde{\phi}(\xi)\) by utilizing Wolfram Mathematica's \texttt{ParametricNDSolve} command \cite{Mathematica}. From Theorem \ref{thm:PE_homoclinic}, we know \(\mu~=~1.5\) meets the constraints on \(\mu\) to have a traveling wave in one spatial dimension. The potential energy function, corresponding homoclinic orbit, and resultant solitary wave for this wave speed is shown in Figure \ref{fig:hot_phase}. We also have that \(\mu~=~1.6\) does not meet the conditions of Theorem \ref{thm:PE_homoclinic}. To be specific, at this wave speed, \(\mathcal{V}(\tilde{\phi}^*) < 0\). The state in phase space is then an incomplete orbit that does not yield a solitary wave, as seen in Figure \ref{fig:hot_phase_broken}.

\begin{figure}[h]
     \centering
     \begin{subfigure}[b]{0.33\textwidth}
         \centering
         \includegraphics[width=\textwidth, trim={0 0 0 15pt}, clip]{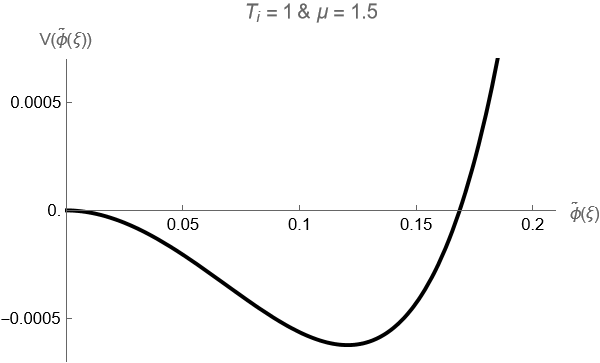}
         \caption{Potential Energy, \(\mathcal{V}(\tilde{\phi})\)}
         \label{fig:PE1}
     \end{subfigure}
     \hfill
     \begin{subfigure}[b]{0.32\textwidth}
         \centering
         \includegraphics[width=\textwidth, trim={0 0 0 15pt}, clip]{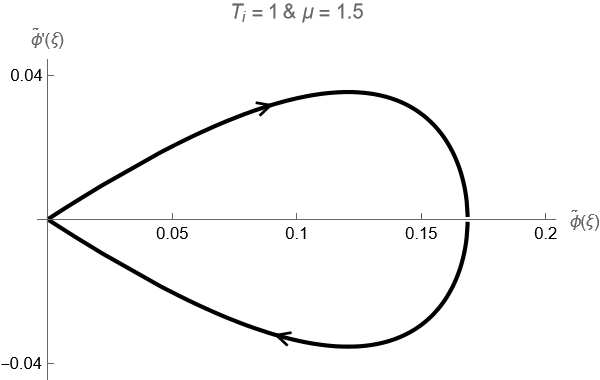}
         \caption{Homoclinic Orbit, \(\tilde{\phi}'(\xi)\)}
         \label{fig:orbit1}
     \end{subfigure}
     \hfill
     \begin{subfigure}[b]{0.32\textwidth}
         \centering
         \includegraphics[width=\textwidth, trim={0 0 0 15pt}, clip]{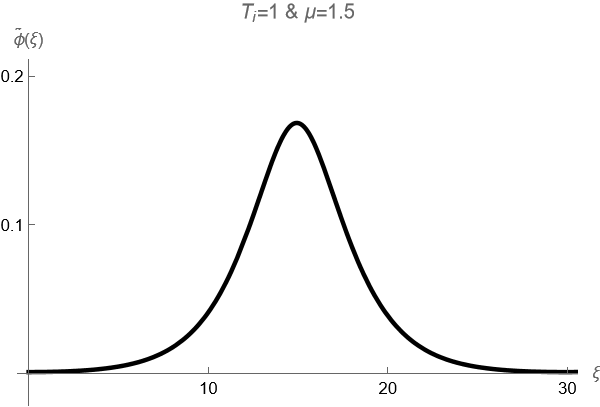}
         \caption{Solitary Wave, \(\tilde{\phi}(\xi)\)}
         \label{fig:wave1}
     \end{subfigure}
        \caption{A hot two-fluid plasma \((\tau_i = 1)\) with wave speed \(\mu =1.5\), utilizing Newton's Method to approximate \(k_\pm(\tilde{\phi}(\xi))\).}
        \label{fig:hot_phase}
\end{figure}

\begin{figure}[h]
     \centering
     \begin{subfigure}[b]{0.33\textwidth}
         \centering
         \includegraphics[width=\textwidth, trim={0 0 0 15pt}, clip]{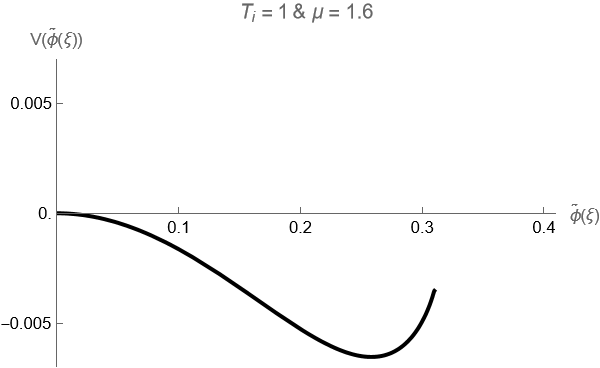}
         \caption{Potential Energy, \(\mathcal{V}(\tilde{\phi})\)}
         \label{fig:PE2}
     \end{subfigure}
     \hfill
     \begin{subfigure}[b]{0.32\textwidth}
         \centering
         \includegraphics[width=\textwidth, trim={0 0 0 15pt}, clip]{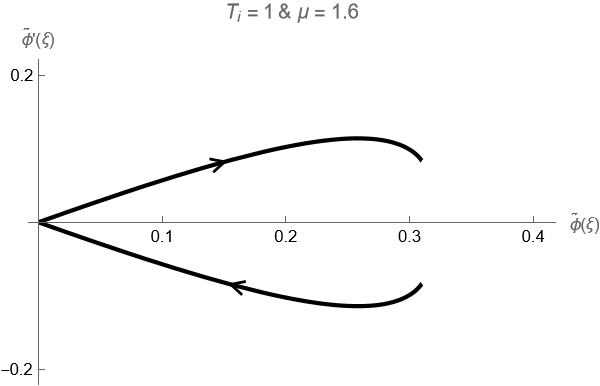}
         \caption{Incomplete Orbit, \(\tilde{\phi}'(\xi)\)}
         \label{fig:orbit2}
     \end{subfigure}
     \hfill
     \begin{subfigure}[b]{0.32\textwidth}
         \centering
         \includegraphics[width=\textwidth, trim={0 0 0 15pt}, clip]{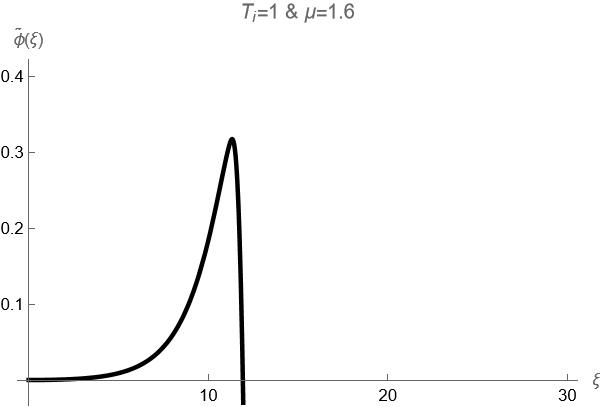}
         \caption{Incomplete Wave, \(\tilde{\phi}(\xi)\)}
         \label{fig:wave2}
     \end{subfigure}
        \caption{A hot two-fluid plasma \((\tau_i = 1)\) with wave speed \(\mu =1.6\), utilizing Newton's Method to approximate \(k_\pm(\tilde{\phi}(\xi))\).}
        \label{fig:hot_phase_broken}
\end{figure}

\subsection{Traveling Waves in a single-fluid plasma}

Haragus and Scheel have already established cold, single-fluid plasmas possess solitary wave solutions \cite{haragus}. However, to our knowledge, hot single-fluid plasmas have not been studied. Thus, we will do so here. We will perform the same process as before to 
\begin{equation}
    \label{eq::3hotrescale}
    \begin{aligned}
        \partial_t n_+ + \partial_x(n_+ v_+) = 0, \qquad
        \partial_t v_+ + \tfrac{1}{2}\partial_x v_+^2 + \tau_i \partial_x \ln(n_+) + \partial_x\phi = 0, \qquad
        \partial_{xx}\phi - e^{\phi} + n_+ = 0,
    \end{aligned}
\end{equation}
the single-fluid Euler-Poisson system for plasma. Implementing a traveling wave ansatz, we obtain results identical to \eqref{eq:travwave_dens_vel} with a positive charge, and \eqref{eq:travwave_phi_vel+}, yielding an equivalent \(k_+(\tilde{\phi})\).
%
The three-component partial differential system becomes the ordinary differential equation
\begin{equation}\label{eq:3travwave_ode}
\tilde{\phi}'' = e^{\tilde{\phi}} - 1 - \frac{k_+(\tilde{\phi})}{\mu-k_+(\tilde{\phi})}.
\end{equation}
We then multiply this by \(\tilde{\phi}'\) and integrate it to procure our potential energy system for a single-fluid plasma. After performing the corresponding analysis, we find wave speeds for which the system has traveling waves. The results are given in Theorem~\ref{thm:3PE_homoclinic} and we provide a pictorial representation in Figure~\ref{fig:3hot_phase} for a hot plasma with speed \(\mu=1.5\).

\begin{center}{\rule{4cm}{0.4pt}}\end{center}

\begin{theorem}\label{thm:3PE_homoclinic}
    The ``potential energy'' system for a single-fluid plasma,
    \begin{align*}
        \frac{(\tilde{\phi}')^2}{2} & = -\mathcal{V}(\tilde{\phi}), \\[5pt]
        \mathcal{V}(\tilde{\phi}) & = -e^{\tilde{\phi}} + \tilde{\phi} - \tau_i \ln (\mu - k_+(\tilde{\phi}) ) + \tfrac{1}{2}(k_+(\tilde{\phi}))^2 - \frac{\mu \tau_i}{\mu - k_+(\tilde{\phi})} + \tau_i (\ln\mu + 1) + 1, \\[5pt]
        \tilde{\phi} & \in \left[-\infty, \, k_+^{-1}(\mu - \sqrt{\tau_i})\right] 
    \end{align*}
    acquires homoclinic orbits for wave speeds \(\mu\) such that
    \begin{align*}
         \mu & > \sqrt{1 +\tau_i} \\
         -e^{\frac{1}{2}\left(\mu^2 - \tau_i\right) + \tau_i \ln\left(\frac{\sqrt{\tau_i}}{\mu}\right)} & + \left(\mu - \sqrt{\tau_i}\right)^2 + 1 > 0
    \end{align*}
\end{theorem}

\begin{proof}
The proof is nearly identical to that of Theorem~\ref{thm:PE_homoclinic} and is therefore omitted.
\end{proof}

\begin{center}{\rule{4cm}{0.4pt}}\end{center}

\begin{figure}[h]
     \centering
     \begin{subfigure}[b]{0.33\textwidth}
         \centering
         \includegraphics[width=\textwidth, trim={0 0 0 15pt}, clip]{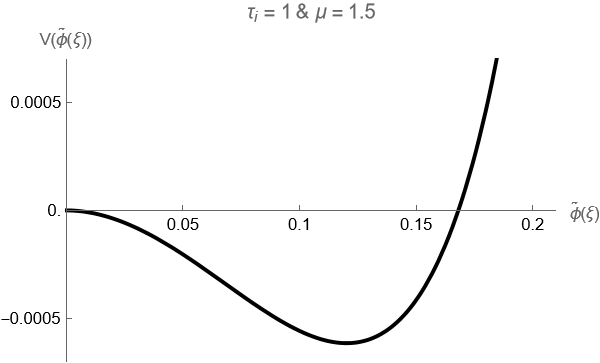}
         \caption{Potential Energy, \(\mathcal{V}(\tilde{\phi})\)}
         \label{fig:PE3}
     \end{subfigure}
     \hfill
     \begin{subfigure}[b]{0.32\textwidth}
         \centering
         \includegraphics[width=\textwidth, trim={0 0 0 15pt}, clip]{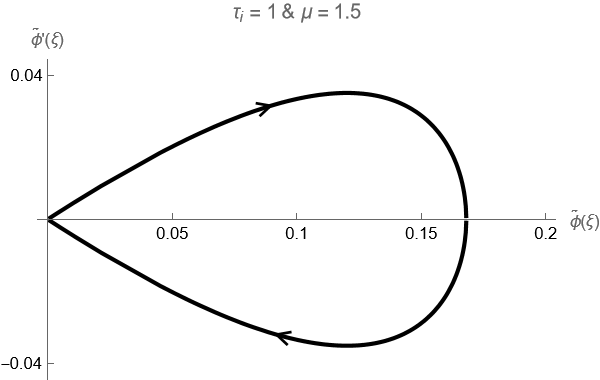}
         \caption{Homoclinic Orbit, \(\tilde{\phi}'(\xi)\)}
         \label{fig:orbit3}
     \end{subfigure}
     \hfill
     \begin{subfigure}[b]{0.32\textwidth}
         \centering
         \includegraphics[width=\textwidth, trim={0 0 0 15pt}, clip]{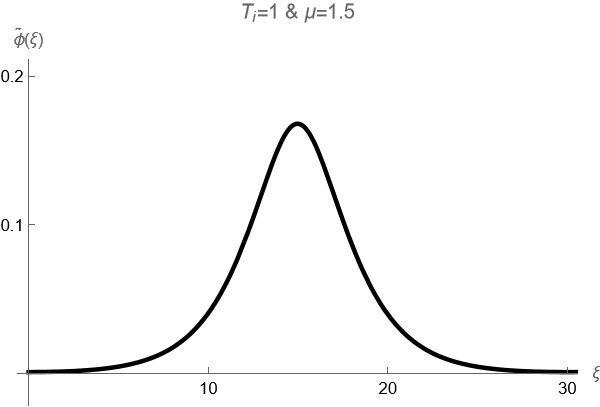}
         \caption{Solitary Wave, \(\tilde{\phi}(\xi)\)}
         \label{fig:wave3}
     \end{subfigure}
        \caption{A hot, single-fluid plasma \((\tau_i = 1, \; m_e=0)\) with wave speed \(\mu =1.5\), utilizing Newton's Method to approximate \(k_+(\tilde{\phi}(\xi))\).}
        \label{fig:3hot_phase}
\end{figure}

\noindent Observe that these results are identical to the two-fluid system if we set \(m_e=0\).
\section{Simulations}\label{simulations}

We will now leverage computer simulations to validate our theoretical results. This approach combines the rigor of mathematical analysis with the practicality of computational methods, allowing us to explore the intricacies of plasma behavior and verify the validity of our theorems. Through this process, we aim to gain new insights that may not be readily apparent through traditional proof techniques alone.

\subsection{The Blueprint}

We begin with assigning \(n_\pm \to 1 + n_\pm\) so all variables are asymptotically null. Separating time and spatial variables, we continue with breaking \eqref{eq::5hotrescale} into its linear and nonlinear pieces. We will also incorporate the linear Taylor expansion of the logarithmic term with coefficient \(\rho\) that depends on the value of \(\tau_i\). The system is
    \begin{equation*}
            \begin{aligned}
                \partial_t n_\pm & = -\partial_x v_\pm -\partial_x(n_\pm v_\pm) ,
                 \\[4pt]
                %
                %
                \partial_t v_+ & = -\rho\partial_x n_+  - \partial_x\phi - \tfrac{1}{2}\partial_x v_+^2 - \tau_i \partial_x \ln(1 + n_+) + \rho\partial_x n_+,
                \\[4pt]
                \partial_t v_- & = -\tfrac{1}{m_e}\partial_x n_- + \tfrac{1}{m_e} \partial_x\phi - \tfrac{1}{2}\partial_x v_-^2 -\tfrac{1}{m_e}\partial_x \ln(1 + n_-) + \tfrac{1}{m_e}\partial_x n_- ,
                 \\[4pt]
                \partial_{xx}\phi & =  n_- - n_+.
            \end{aligned}
        \end{equation*}

To make it easier to solve, we will use four variables instead of five. Define 
\begin{equation}\label{sim_E_S_def}
    \partial_x E \coloneqq n_- - n_+, \quad S \coloneqq n_- + n_+.
\end{equation}
Then 
\begin{equation}\label{sim_dens_ES_def}
    n_+ = \tfrac{1}{2}(S - \partial_x E), \quad n_- = \tfrac{1}{2}(S + \partial_x E),
\end{equation}
and our new Euler-Poisson system is 
\begin{subequations}
    \label{eq:EP_SE_sys}
    \begin{align}
    \label{eq:dtS}
        \partial_t S & = -\partial_x \left(v_- + v_+\right) - \tfrac{1}{2}\partial_x\left(\left(S + \partial_x E \right)v_- + \left(S - \partial_x E \right)v_+ \right) ,
        \\[8pt]
    \label{eq:dtE}
        \partial_t E & = -v_- + v_+ - \tfrac{1}{2}\left(\left(S + \partial_x E \right)v_- - \left(S - \partial_x E \right)v_+ \right),
        \\[8pt]
    \label{eq:dtv+}
        \partial_t v_+ & = -\tfrac{\rho}{2}\partial_x S + \left(\tfrac{\rho}{2} \partial_x^2 - 1\right)E - \partial_x\left(\tfrac{1}{2} v_+^2 + \tau_i \ln\left(1+ \tfrac{S - \partial_x E}{2} \right) - \tfrac{\rho}{2} \left(S - \partial_x E\right) \right),
        \\[8pt]
    \label{eq:dtv-}
        \partial_t v_- & =\tfrac{1}{m_e}\left( - \tfrac{1}{2}\partial_x S - \left(\tfrac{1}{2}\partial_x^2 - 1\right)E\right)  - \tfrac{1}{m_e}\partial_x \left(\tfrac{m_e}{2} v_-^2 + \ln\left(1+ \tfrac{S + \partial_x E}{2}\right) - \tfrac{1}{2} \left(S + \partial_x E\right) \right).
    \end{align}
\end{subequations}
We will integrate each equation with respect to time, \(t\), from \(t_1\) to \(t_2\). Regarding the evaluation at these boundary points, variables are given a corresponding time subscript: \(var(x,t_2)~\coloneqq~var_2\), \(var(x,t_1)~\coloneqq~var_1\).

The integration yields
    \begin{align*}
        S_2 - S_1 & = \int_{t_1}^{t_2} \left(-\partial_x \left(v_- + v_+\right) - \tfrac{1}{2}\partial_x\left(\left(S + \partial_x E \right)v_- + \left(S - \partial_x E \right)v_+ \right) \right) dt, \\[4pt]
        E_2 - E_1 & = \int_{t_1}^{t_2} \left( -v_- + v_+ - \tfrac{1}{2}\left(\left(S + \partial_x E \right)v_- - \left(S - \partial_x E \right)v_+ \right) \right) dt,
        \\[4pt]
        v_{+,2} - v_{+,1} & = \int_{t_1}^{t_2} \left(\right. -\tfrac{\rho}{2}\partial_x S + \left(\tfrac{\rho}{2} \partial_x^2 - 1\right)E \\
        & \qquad\qquad - \partial_x\left(\tfrac{1}{2} v_+^2 + \tau_i \ln\left(1+ \tfrac{S - \partial_x E}{2} \right) - \tfrac{\rho}{2} \left(S - \partial_x E\right) \right) \left.\right)dt,
        \\[4pt]
        v_{-,2} - v_{-,1} & = \int_{t_1}^{t_2} \left(\right. \tfrac{1}{m_e}\left(- \tfrac{1}{2}\partial_x S - \left(\tfrac{1}{2}\partial_x^2 - 1\right)E \right) \\
        & \qquad\qquad - \tfrac{1}{m_e}\partial_x \left(\tfrac{m_e}{2} v_-^2 + \ln\left(1+ \tfrac{S + \partial_x E}{2}\right) - \tfrac{1}{2} \left(S + \partial_x E\right) \right) \left.\right) dt .
    \end{align*}
While the integrands containing time derivatives were straightforward, the pieces with spatial derivatives will need an approximation. Define \(D \coloneqq \partial_x\) and \(\Delta t := t_2 - t_1\). Then we apply the Leibniz Integral Rule and Trapezoidal rule,
\begin{equation*}
\int_{t_1}^{t_2} \partial_x var(x,t) dt = \partial_x \left( \int_{t_1}^{t_2} var(x,t) dt\right) \approx \tfrac{D \Delta t}{2}\left(var_2 + var_1 \right).
\end{equation*}
Following, our system of integrals becomes
    \begin{align*}
        S_2 - S_1 & = \tfrac{\Delta t}{2} \left(-D\left(v_{-,2} + v_{+,2}\right) - \tfrac{1}{2}D\left(\left(S_2 + D E_2 \right) v_{-,2} + \left(S_2 - D E_2 \right)v_{+,2}\right)\right)\\
        & \, + \tfrac{\Delta t}{2} \left(-D\left(v_{-,1} + v_{+,1}\right) - \tfrac{1}{2}D\left(\left(S_1 + D E_1 \right) v_{-,1} + \left(S_1 - D E_1 \right)v_{+,1}\right) \right) ,
        \\[8pt]
        E_2 - E_1 & = \tfrac{\Delta t}{2} \left( -v_{-,2} + v_{+,2} - \tfrac{1}{2}\left(\left(S_2 + D E_2 \right) v_{-,2} - \left(S_2 - D E_2 \right)v_{+,2} \right) \right)
        \\ 
        & \, + \tfrac{\Delta t}{2} \left( -v_{-,1} + v_{+,1} - \tfrac{1}{2}\left(\left(S_1 + D E_1 \right) v_{-,1} - \left(S_1 - D E_1 \right)v_{+,1} \right) \right) ,
        \\[8pt]
        v_{+,2} - v_{+,1} & = \tfrac{\Delta t}{2} \left(\right.-\tfrac{\rho}{2}D S_2 + \left(\tfrac{\rho}{2} D^2 - 1\right)E_2 \\
        & \qquad\quad\, - D\left(\tfrac{1}{2} v_{+,2}^2 + \tau_i \ln(1+ \tfrac{S_2 - D E_2}{2} ) - \tfrac{\rho}{2} \left(S_2 - D E_2\right) \right) \left.\right)
        \\ 
        & \, + \tfrac{\Delta t}{2} \left(\right.-\tfrac{\rho}{2}D S_1 + \left(\tfrac{\rho}{2} D^2 - 1\right)E_1 \\
        & \qquad\quad\, - D\left(\tfrac{1}{2} v_{+,1}^2 + \tau_i \ln(1+ \tfrac{S_1 - D E_1}{2} ) - \tfrac{\rho}{2} \left(S_1 - D E_1\right) \right) \left.\right) ,
        \\[8pt]
        v_{-,2} - v_{-,1} & = \tfrac{\Delta t}{2}\left(\right. \tfrac{1}{m_e}\left(- \tfrac{1}{2} D S_2 - \left(\tfrac{1}{2}D^2 - 1\right)E_2\right) \\
        & \qquad\quad\, - \tfrac{1}{m_e}D\left(\tfrac{m_e}{2} v_{-,2}^2 + \ln(1+ \tfrac{S_2 + D E_2}{2}) - \tfrac{1}{2} \left(S_2 + D E_2\right) \right) \left.\right) \\
        & \, + \tfrac{\Delta t}{2}\left(\right. \tfrac{1}{m_e}\left(-\tfrac{1}{2} D S_1 - \left(\tfrac{1}{2}D^2 - 1 \right)E_1 \right) \\
        & \qquad\quad\, - \tfrac{1}{m_e}D\left(\tfrac{m_e}{2} v_{-,1}^2 + \ln(1+ \tfrac{S_1 + D E_1}{2}) - \tfrac{1}{2} \left(S_1 + D E_1\right) \right) \left.\right) .
    \end{align*}

We will move the future time steps to the left side of the equations and put the current on the right, keeping the nonlinear pieces on the right-hand side as well. Written as a system of matrices, this is
\begin{align}
& \nonumber 
\begin{bmatrix}
1 & 0 & \tfrac{1}{2} D\Delta t & \tfrac{1}{2} D\Delta t \\[4pt]
0 & 1 & -\tfrac{1}{2} \Delta t & \tfrac{1}{2} \Delta t \\[4pt]
\tfrac{\rho}{4} D\Delta t & \tfrac{1}{2}\Delta t\left(1 - \tfrac{\rho}{2} D^2 \right) & 1 & 0 \\[4pt]
\tfrac{1}{4m_e} D\Delta t & -\tfrac{1}{2m_e}\Delta t\left(1 - \tfrac{1}{2} D^2 \right) & 0 & 1 \\[4pt]
\end{bmatrix}
\begin{bmatrix}
S_2 \\[4pt]
E_2 \\[4pt]
v_{+,2} \\[4pt]
v_{-,2} \\[4pt]
\end{bmatrix} \\[10pt]
& \nonumber  =
\begin{bmatrix}
1 & 0 & -\tfrac{1}{2} D\Delta t & -\tfrac{1}{2} D\Delta t \\[4pt]
0 & 1 & \tfrac{1}{2} \Delta t & -\tfrac{1}{2} \Delta t \\[4pt]
-\tfrac{\rho}{4} D\Delta t & -\tfrac{1}{2}\Delta t\left(1 - \tfrac{\rho}{2} D^2 \right) & 1 & 0 \\[4pt]
-\tfrac{1}{4m_e} D\Delta t & \tfrac{1}{2m_e}\Delta t\left(1 - \tfrac{1}{2} D^2 \right) & 0 & 1 \\[4pt]
\end{bmatrix}
\begin{bmatrix}
S_1 \\[4pt]
E_1 \\[4pt]
v_{+,1} \\[4pt]
v_{-,1} \\[4pt]
\end{bmatrix} \\[10pt]
\nonumber & \quad +
\begin{bmatrix}
- \tfrac{1}{4} D\Delta t\left( \left(S_2 + D E_2\right) v_{-,2} + \left(S_2 - D E_2\right) v_{+,2} \right) \\[4pt]
-\tfrac{1}{4} \Delta t \left(\left(S_2 + D E_2\right) v_{-,2} - \left(S_2 - D E_2\right) v_{+,2} \right) \\[4pt]
- \tfrac{1}{2} D\Delta t \left(\tfrac{1}{2} v_{+,2}^2 + \tau_i \ln(1+ \tfrac{S_2 - D E_2}{2} ) - \tfrac{\rho}{2} \left( S_2 - D E_2 \right) \right) \\[4pt]
-\tfrac{1}{2 m_e} D\Delta t \left(\tfrac{m_e}{2} v_{-,2}^2 + \ln(1+ \tfrac{S_2 + D E_2}{2}) - \tfrac{1}{2}\left(S_2 + D E_2 \right) \right) \\[4pt]
\end{bmatrix} \nonumber \\[10pt]
& \quad +
\begin{bmatrix}
-\tfrac{1}{4} D\Delta t \left( \left(S_1 + D E_1\right) v_{-,1} + \left(S_1 - D E_1\right) v_{+,1} \right) \\[4pt]
-\tfrac{1}{4} \Delta t \left(\left(S_1 + D E_1\right) v_{-,1} - \left(S_1 - D E_1\right) v_{+,1} \right) \\[4pt]
-\tfrac{1}{2} D\Delta t \left(\tfrac{1}{2} v_{+,1}^2 + \tau_i \ln(1+ \tfrac{S_1 - D E_1}{2} ) - \tfrac{\rho}{2} \left( S_1 - D E_1 \right) \right) \\[4pt]
-\tfrac{1}{2m_e} D\Delta t \left(\tfrac{m_e}{2} v_{-,1}^2 + \ln(1+ \tfrac{S_1 + D E_1}{2}) - \tfrac{1}{2}\left(S_1 + D E_1 \right) \right) \\[4pt]
\end{bmatrix} \label{eq:matrix_code}
\end{align}
This system of matrices is what we will use to solve \eqref{eq:EP_SE_sys} and simulate \eqref{eq::5hotrescale}. Based on our trials, we find the best choice for \(\rho\) to be 
\begin{equation*}
    \rho =
    \begin{cases}
        1 & \tau_i = 1, \\
        1/90 & \tau_i = 0.
    \end{cases}
\end{equation*}

\subsection{Too Fast Too Fourier}

Based on the work of Sattinger and Li, we write an algorithm in \texttt{python} to solve our system iteratively \cite{li1999soliton}. To perform the algebraic manipulation, and handle the derivatives, we utilize Fast Fourier Transformations. Moreover, we will employ an averaging technique to manage the nonlinear terms. The entirety of our computation was coded from scratch and only used the basic math packages \texttt{math} and \texttt{numpy} \cite{van1995python, harris2020array}. The process follows.

Our equation of matrices \eqref{eq:matrix_code} is essentially \(A \mathbf{w}_2 = B \mathbf{w}_1 + \mathbf{q}_2 + \mathbf{q}_1\). Setting \(D=i\omega\) (\(\omega\) is the angular frequency), we find the matrix \(A\) to be invertible, so that 
\begin{equation}\label{eq::future_time_calc}
    \mathbf{w}_2 = A^{-1}B \mathbf{w}_1 + A^{-1}\mathbf{q}_2 + A^{-1}\mathbf{q}_1.
\end{equation}
Thus, we can find the future time step \(\mathbf{w}_2\) by performing some arithmetic on the present time step \(\mathbf{w}_1\). However, the nonlinear future step \(\mathbf{q}_2\) throws a wrench in our calculations. Adapting, we will average our nonlinear terms \(\mathbf{q}_1\) and \(\mathbf{q}_2\), as stated previously. We begin with \(\mathbf{q}_2 = \mathbf{q}_1\), then obtain a better estimation for \(\mathbf{q}_2\) by going through a few smaller sub-calculations following the \eqref{eq::future_time_calc} formula. A pseudo-code demonstrating this method reads
\begin{lstlisting}[language=Python]
# Commands
fft = Fast Fourier Transform
ifft = Inverse Fast Fourier Transform
# Set up empty matrix
w = ZeroMatrix(Spatial Grid Size, Temporal Grid Size)
# Insert initial condition
w[:,0] = w0
# Solve future time steps
for t in range(1, Temporal Grid Size):
    fftq1 = fft(q1 calc with w[:,t-1])
    wt1 = (A^{-1} x B)*fft(w[:,t-1]) + A^{-1}*fftq1
    fftq2 = fftq1
    for step in range(1,5):
        wt2 = wt1 + A^{-1}*fftq2
        w2 = ifft(wt2)
        fftq2 = fft(q2 calc with w2)
    fftw = wt1 + A^{-1}*fftq2
    w[:,t] = ifft(fftw)
\end{lstlisting}

With this method, we need only plug in an initial condition to find solutions \((S, \, E, \, v_+, \, v_-)\) for \eqref{eq:EP_SE_sys} at each time-step. Once we have these, we can obtain answers for the other variables \(n_+\), \(n_-\), and \(\phi\). The former two are acquired using \eqref{sim_dens_ES_def}, which notably requires \(\partial_x E\). The spatial derivative of \(E\) at each time step is calculated using finite differences, 
\begin{equation*}
    \partial_x E(x,t) = \frac{E(x, t) - E(x-\Delta x,t)}{\Delta x}.
\end{equation*}
%
The only nuanced piece to use this formula is enforcing
\begin{equation*}
    E(L + \Delta x, t) = E(-L,t).
\end{equation*}
Furthermore, because \(\partial_x \phi = E\), we can recover \(\phi(x, t)\) by performing the calculation,
\begin{equation*}
    \texttt{fft}(E) = i\omega \texttt{fft}(\phi) \quad\implies\quad \phi = \texttt{ifft}\left(\frac{\texttt{fft}(E)}{i\omega}\right).
\end{equation*}

\subsection{A Traveling Wave Initial Condition}

We will now put our algorithm to work and look at the evolution of \eqref{eq::5hotrescale} over time and space. For the initial condition, we first use a traveling wave with speed \(\mu\) that satisfies Theorem~\ref{thm:PE_homoclinic}. Note that because this is an initial condition, i.e. \(t=0\), \(\xi\) and \(x\) are the same. Define 
\begin{align*}
    \phi(x,0) & = \tilde{\phi}(\xi) , \\
    v_\pm(x,0) & = \tilde{v}_\pm(\xi) = k_\pm (\tilde{\phi}(\xi)), \\
    n_\pm(x,0) & = \tilde{n}_\pm(\xi) = \frac{\tilde{v}_\pm(\xi)}{\mu - \tilde{v}_\pm(\xi)}.
\end{align*}
Then from \eqref{sim_E_S_def}, 
\begin{align*}
    S(x,0) & = n_-(x,0) + n_+(x,0), \\
    E(x,0) & = \partial_x \phi(x,0) = \frac{\phi(x, 0) - \phi(x-\Delta x,0)}{\Delta x}.
\end{align*}

After setting up all of our initial conditions in \texttt{Wolfram Mathematica}, we can import them into \texttt{python} as an array. We need to extend the data, furthermore, because it must fit our spatial domain, \(x\in[-L, L]\), if \(L > \texttt{Max}(\xi)\). To ensure the initial condition has the same length as the spatial domain, we simply pad the array with zeros on either side since the traveling wave should go to zero asymptotically. However, due to small computational discrepancies, the original padding produced jumps in the tails of the data as seen in Figure~\ref{fig:icVec}. These tiny cracks led to inaccurate behavior in our simulations. So, we remedied this issue by convolving the initial conditions with the mollifier
\begin{equation*}
m(x) =
    \begin{cases}
        \kappa e^{-1/(1-|x|^2)} & |x| < 1, \\
        0 & else,
    \end{cases}
\end{equation*}
containing a specific scaling \(\kappa\) to ensure normalization. After convolving, we obtained the desired smoother initial conditions to use in our simulations. Our mollified example is shown in Figure~\ref{fig:icMolli}.
\begin{figure}[h!]
     \centering
     \begin{subfigure}[b]{0.45\textwidth}
         \centering
         \includegraphics[width=\textwidth, trim={55pt 35pt 35pt 275pt}, clip]{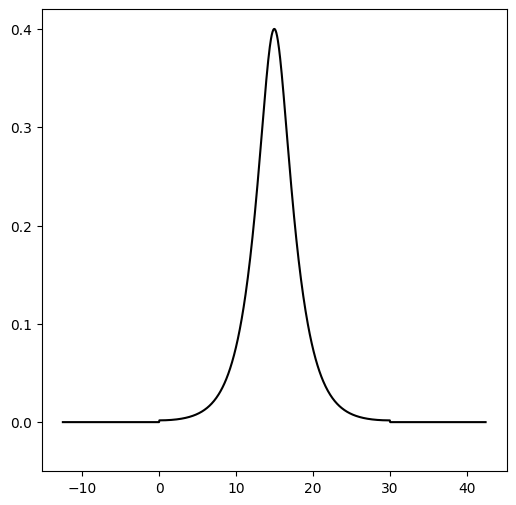} 
         \caption{Tails of initial import}
         \label{fig:icVec}
     \end{subfigure}
     \hfill
     \begin{subfigure}[b]{0.45\textwidth}
         \centering
         \includegraphics[width=\textwidth, trim={55pt 35pt 35pt 275pt}, clip]{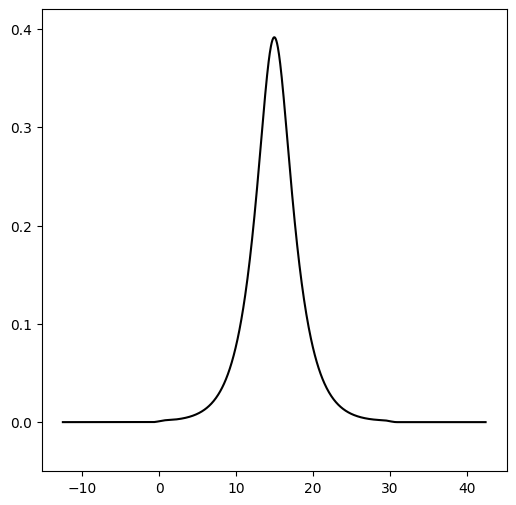}
         \caption{Tails post-mollification}
         \label{fig:icMolli}
     \end{subfigure}
     \caption{Mollifying the initial condition \(S(x,0)\).}
        \label{fig:mollifier}
\end{figure}

Now that we have established the methodologies to simulate our system, we can dive into our findings. We will observe how solutions to \eqref{eq::5hotrescale} evolve across space and time, gaining valuable insights into the intricate interplay between the particles and the resulting fluid dynamics. We provide results for each particle group's velocity and density, along with the electric potential, in both hot and cold plasmas. The initial condition for a hot plasma has a wave speed of \(\mu=1.5\), while for cold plasma, \(\mu = 1.154\). We take the domain size \(L=250\), spatial step-size \(\Delta x = 0.05\), temporal step-size \(\Delta t = 0.5\), and max time \(T=100\). 

\subsubsection{Particle Velocities}

We start with the particle velocities as they display the most fascinating phenomena. The velocity will tell us how individual particles move within our medium. Recall their governing equations are directly impacted by our key parameters, \(\tau_i\) and \(m_e\). We will use their varying values to observe how they affect the Euler-Poisson system as a whole. 

For a hot plasma, \(\tau_i=1\), the evolution of the ion velocity is shown in Figure~\ref{fig:waterfall_hot_v+}, and electron velocity in Figure~\ref{fig:waterfall_hot_v-}. Both particles exhibit the initial velocity wave propagating forward in space. As time progresses, this wave leaves a trail of dispersive waves in its wake. The children follow their mother, growing and expanding in her footsteps, a symbolic echo of the physical behavior observed. Unlike the ion velocity, however, electron velocity exhibits an extreme warbling centered at the same spatial interval as the initial wave. This discrepancy is likely due to the singular perturbation by \(m_e\) in \eqref{eq::5hotrescale}.

The warbling contains two major entities that resemble a mountain with a valley connected to a canyon with a hill. With each time step, their positions flip, creating a back-and-forth motion. As this warbling continues, the region connecting the mountain and canyon expands. Initially flat, this expansion steepens, pushing outward and reducing the disparity between the depth and altitudes. It erodes the valleys into canyons, and deposits sediment, transforming the hill into a mountain. This process sends portions outward to form the adjacent hill and valley. The middle part flattens again before widening once more, and the cycle repeats. Slowly, the overall expanse of the warbling widens. The cardiologist's worst nightmare is shown in Figure~\ref{fig:velocity_oscilations}.
\begin{figure}[h]
    \centering
    \includegraphics[width=\textwidth]{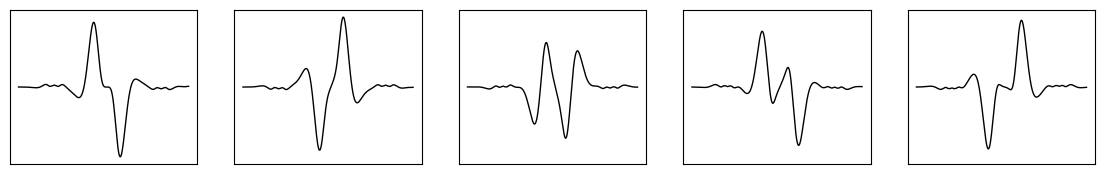}
    \caption{Oscillations in the electron velocity of a hot plasma at times \(t=31.5, \, 34, \, 35.5, \, 37, \, 39.5\).}
    \label{fig:velocity_oscilations}
\end{figure}

This phenomenon is not unique to hot two-fluid plasmas. Indeed, the cold electron velocity in Figure~\ref{fig:waterfall_cold_v-} presents it as well. Moreover, all variables in both thermal regimes exhibit the emergence of wave packets that propagate backward in space, contrasting with the more common direction. It may be easier to observe in the zoomed-in portion of the ion density diagrams, Figures~\ref{fig:waterfall_hot_n+} and \ref{fig:waterfall_cold_n+}. These wavelets in hot plasmas start small and appear to, fairly consistently, grow in size as we follow them backward. However, cold plasma wavelets start strong, shrink, and grow again. This behavior is likely due to destructive wave interference in this region at thermodynamic equilibrium. In other respects, the particularly notable differences between hot and cold particle velocities lie within the overall wave amplitude and magnitude of subsequent behavior. 

%
%
\subsubsection{Particle Densities}

Figures~\ref{fig:waterfall_hot_n+} and \ref{fig:waterfall_hot_n-} illustrate the progression of ion and electron density within a hot plasma, while Figures~\ref{fig:waterfall_cold_n+} and \ref{fig:waterfall_cold_n-} depict the equivalent densities in a cold plasma. These figures reveal a non-uniform distribution of particles, with regions of varying density. Similar to the behavior of particle velocities, we observe the initial density wave moving forward in space, accompanied by smaller trailing waves in both thermal cases. The electron densities, moreover, exhibit remnants of the warbling seen in their velocity, including the widening effect but maintenance of a small amplitude. 

\subsubsection{Electric Potential}

The electric potential highlights the nature of ion and electron movements within the system, with each particle's motion influencing its neighbors and contributing to the overall wave dynamics. This variable's evolution is depicted in Figures~\ref{fig:waterfall_hot_phi} and \ref{fig:waterfall_cold_phi} for hot and cold plasmas, respectively. As we know, the electric potential is governed by the particle densities in the plasma. Hence, we can see their spatial and temporal characteristics combine to form the electric potential waves. The collective behavior includes an initial propagating wave, trailing waves, and symmetric oscillations about the initial wave region.

\begin{figure}
     \centering
     \begin{subfigure}[b]{0.49\textwidth}
         \centering
         \includegraphics[width=\textwidth, trim={0 25pt 0 50pt}, clip]{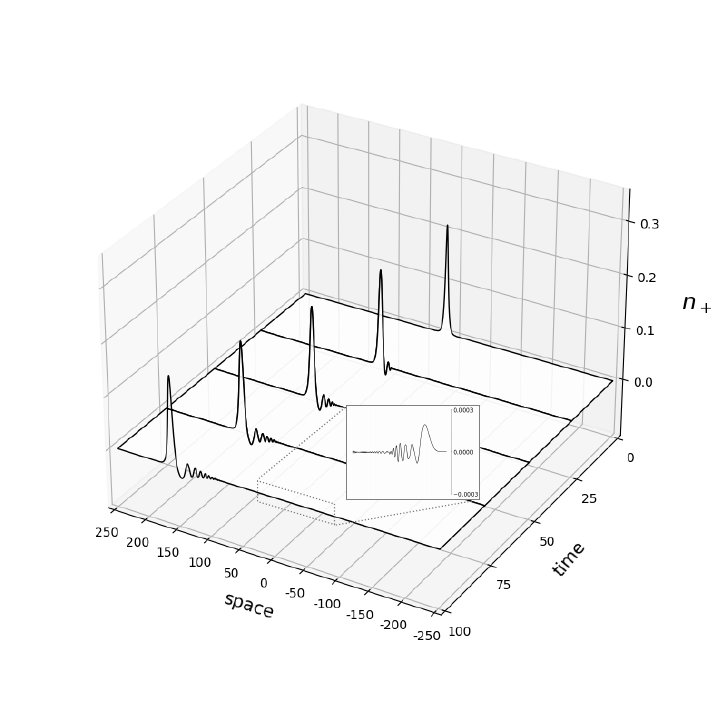}
         \caption{Ion Density, \(n_+(x,t)\) \\[14pt]}
         \label{fig:waterfall_hot_n+}
     \end{subfigure}
     \hfill
     \begin{subfigure}[b]{0.49\textwidth}
         \centering
         \includegraphics[width=\textwidth, trim={0 40pt 0 70pt}, clip]{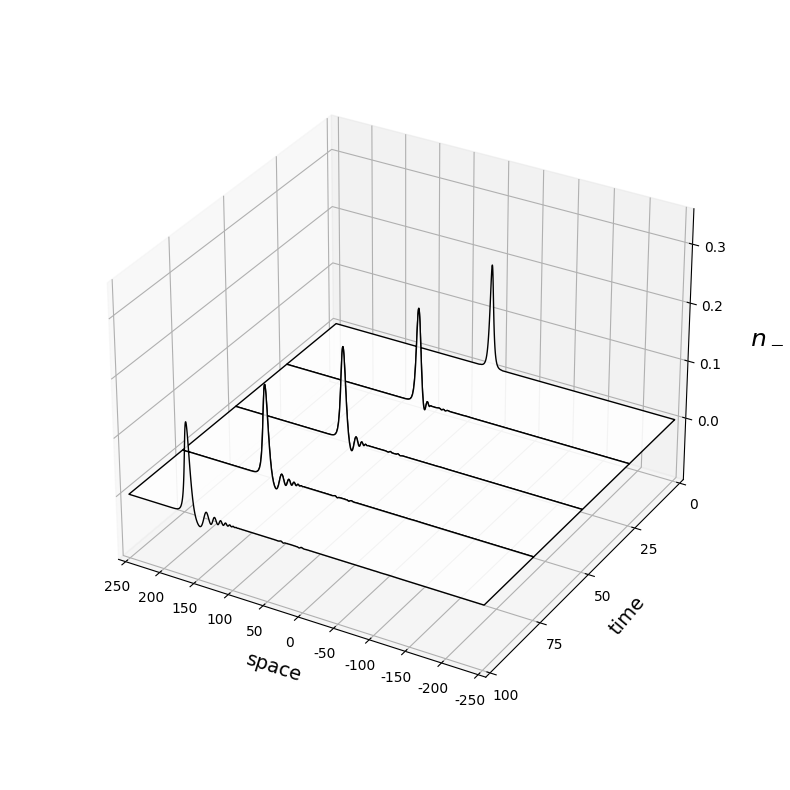}
         \caption{Electron Density, \(n_-(x,t)\) \\[14pt]}
         \label{fig:waterfall_hot_n-}
     \end{subfigure}
     \hfill
     \begin{subfigure}[b]{0.49\textwidth}
         \centering
         \includegraphics[width=\textwidth, trim={0 40pt 0 70pt}, clip]{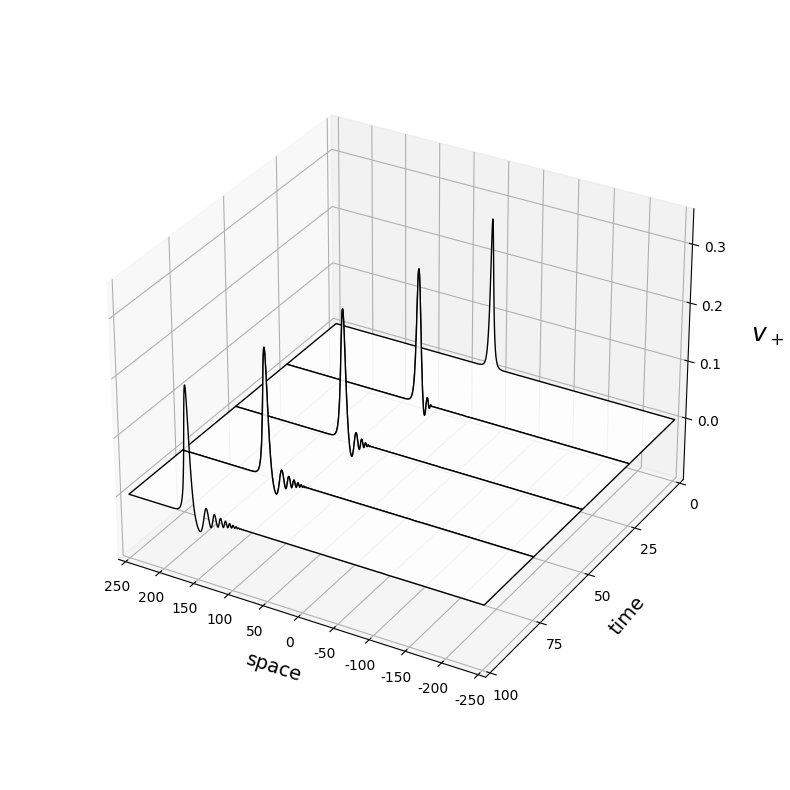}
         \caption{Ion Velocity, \(v_+(x,t)\) \\[14pt]}
         \label{fig:waterfall_hot_v+}
     \end{subfigure}
     \hfill
     \begin{subfigure}[b]{0.49\textwidth}
         \centering
         \includegraphics[width=\textwidth, trim={0 40pt 0 70pt}, clip]{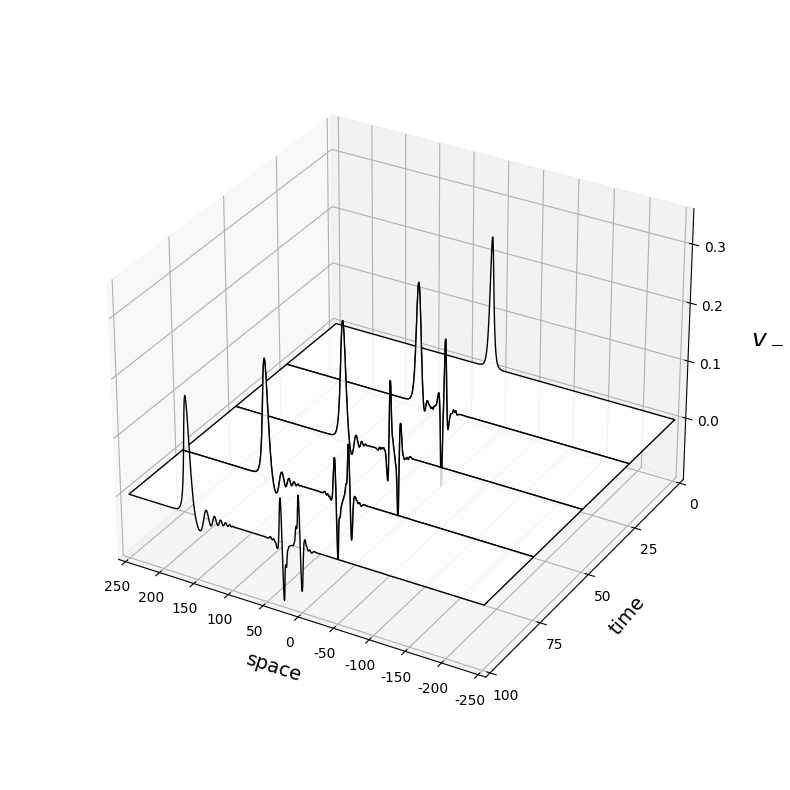}
         \caption{Electron Velocity, \(v_-(x,t)\) \\[14pt]}
         \label{fig:waterfall_hot_v-}
     \end{subfigure}
     \hfill
     \begin{subfigure}[b]{0.49\textwidth}
         \centering
         \includegraphics[width=\textwidth, trim={0 40pt 0 70pt}, clip]{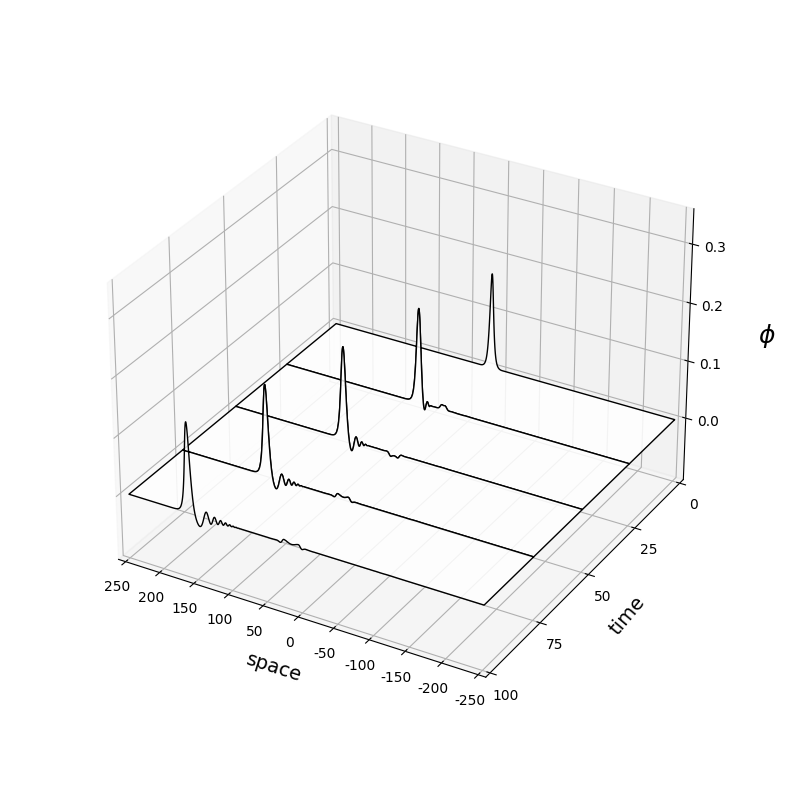}
         \caption{Electric Potential, \(\phi(x,t)\)}
         \label{fig:waterfall_hot_phi}
     \end{subfigure}
     \caption{Propagation of a hot two-fluid plasma (\(\tau_i=1\)) over space and time. The initial condition is a traveling wave solution of speed \(\mu = 1.5\).}
    \label{fig:waterfall_hot}
\end{figure}

\begin{figure}
     \centering
     \begin{subfigure}[b]{0.49\textwidth}
         \centering
         \includegraphics[width=\textwidth, trim={0 25pt 0 50pt}, clip]{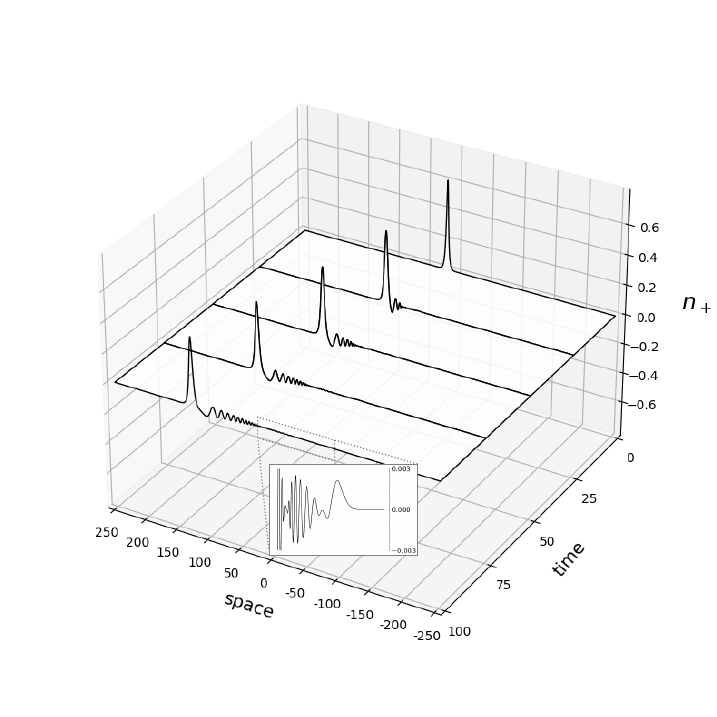}
         \caption{Ion Density, \(n_+(x,t)\) \\[14pt]}
         \label{fig:waterfall_cold_n+}
     \end{subfigure}
     \hfill
     \begin{subfigure}[b]{0.49\textwidth}
         \centering
         \includegraphics[width=\textwidth, trim={0 40pt 0 70pt}, clip]{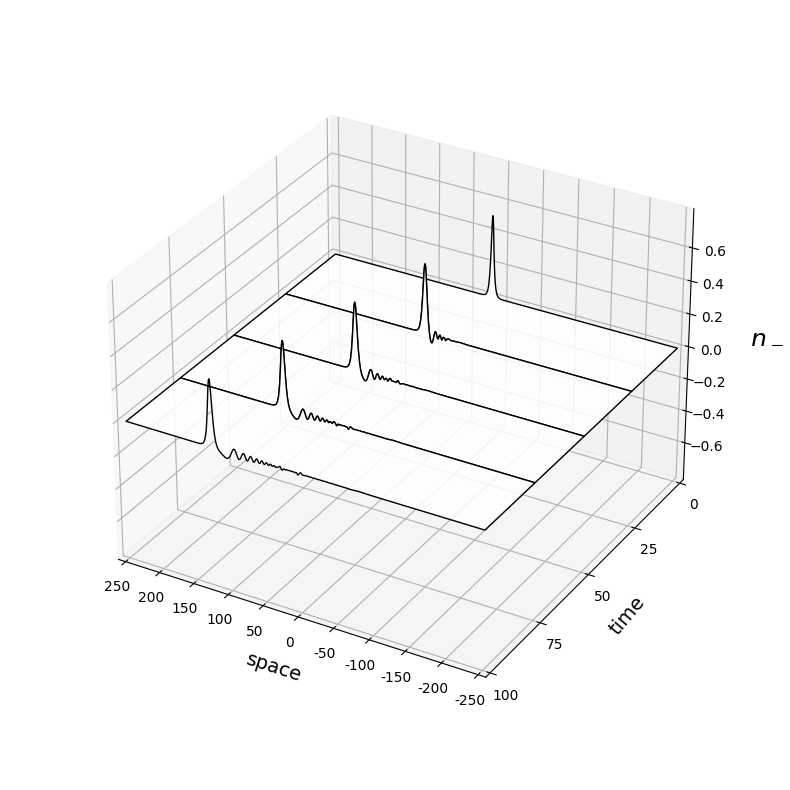}
         \caption{Electron Density, \(n_-(x,t)\) \\[14pt]}
         \label{fig:waterfall_cold_n-}
     \end{subfigure}
     \hfill
     \begin{subfigure}[b]{0.49\textwidth}
         \centering
         \includegraphics[width=\textwidth, trim={0 40pt 0 70pt}, clip]{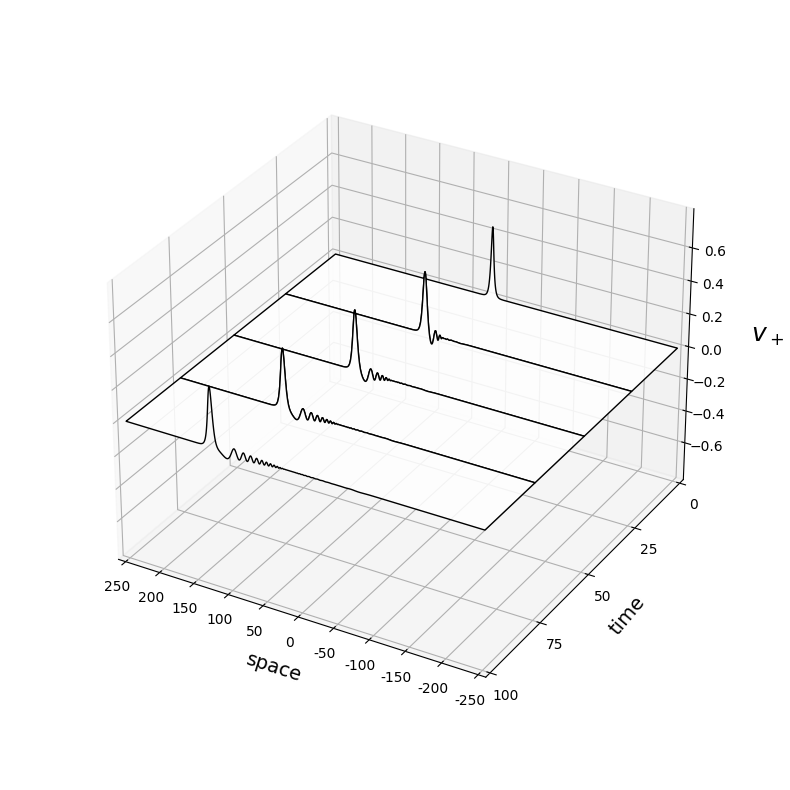}
         \caption{Ion Velocity, \(v_+(x,t)\) \\[14pt]}
         \label{fig:waterfall_cold_v+}
     \end{subfigure}
     \hfill
     \begin{subfigure}[b]{0.49\textwidth}
         \centering
         \includegraphics[width=\textwidth, trim={0 40pt 0 70pt}, clip]{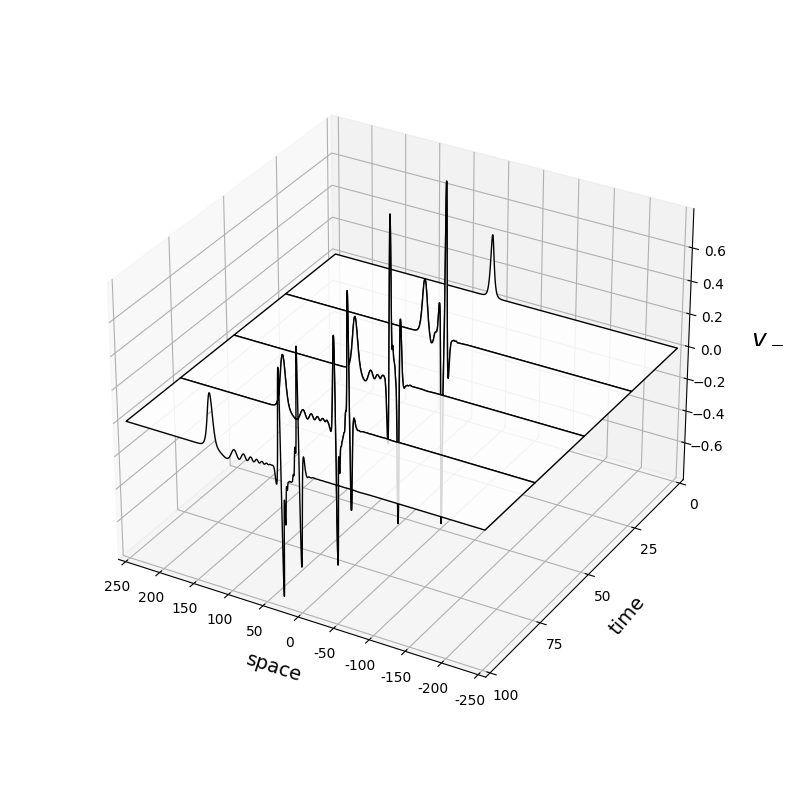}
         \caption{Electron Velocity, \(v_-(x,t)\) \\[14pt]}
         \label{fig:waterfall_cold_v-}
     \end{subfigure}
     \hfill
     \begin{subfigure}[b]{0.49\textwidth}
         \centering
         \includegraphics[width=\textwidth, trim={0 40pt 0 70pt}, clip]{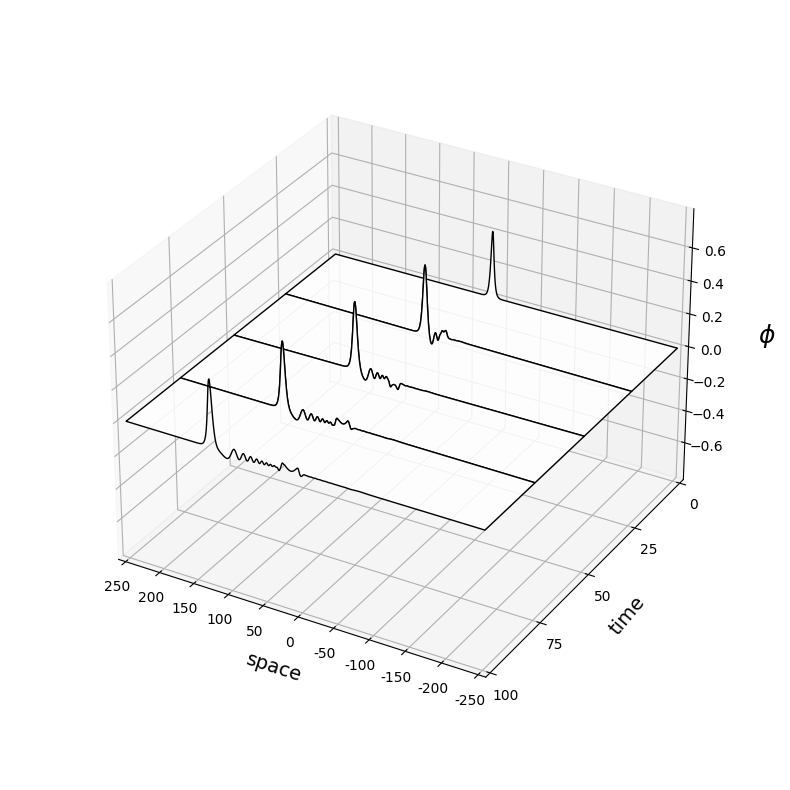}
         \caption{Electric Potential, \(\phi(x,t)\)}
         \label{fig:waterfall_cold_phi}
     \end{subfigure}
     \caption{Propagation of a cold two-fluid plasma (\(\tau_i=0\)) over space and time. The initial condition is a traveling wave solution of speed \(\mu~=~1.154\).}
    \label{fig:waterfall_cold}
\end{figure}

\newpage
\subsection{A KdV Initial-Condition}

The next initial condition we consider comes from the solution to the classic Korteweg de-Vries equation, for which our variables have a first-order approximation \cite{kelting}. The solution, extended by Kelting and Wright, is 
\begin{equation}\label{eq::kdvSol}
    n_+(x,t) = \varepsilon n_+^{(1)}(\varepsilon^{1/2}x,t) = \frac{3\tilde\mu\varepsilon}{c} \text{sech}^2\left((x - (c + \tilde\mu \varepsilon)t)\sqrt{\frac{c \tilde\mu\varepsilon \left(m_e+1\right)}{2\left(c^2 m_e-1\right)^2}}\right)
\end{equation}
where \(\tilde\mu\) is the speed of the wave, \(\varepsilon\) is the amplitude of the initial disturbance, and \(c\) is the speed of sound with
\(
    c=\sqrt{\tfrac{1+\tau_i}{1+m_e}}.
\)
Then from the first-order relation \((\mathcal{S}_1)\) in \cite{kelting}, we get the initial conditions for our other variables,
\begin{align*}
    n_+(x,0) & = \frac{3\tilde\mu\varepsilon}{c} \text{sech}^2\left(x\sqrt{\tfrac{c \tilde\mu\varepsilon \left(m_e+1\right)}{2\left(c^2 m_e-1\right)^2}}\right) \\
    n_-(x,0) & = n_+(x,0), \\
    v_\pm (x,0) & = c n_+(x,0), \\
    \phi(x,0) & = (1-\tau_i m_e)n_+(x,0).
\end{align*}
Thus, \(S(x,0) = 2n_+(x,0)\) and \(E(x,0) = \tfrac{d}{dx}\phi(x,0)\). Furthermore, we keep parameters \(L\), \(\Delta x\), \(\Delta t\), and \(T\) to be the same as the previous simulation.

\subsubsection{The Results}

With the initial condition solving the inhomogeneous KdV equation, we take \(\tilde\mu = 1.5\) for hot plasmas, \(\tilde\mu=1.154\) for cold, and \(\varepsilon = 10^{-2}\). The results for a hot plasma are given in Figure~\ref{fig:waterfall_hot_KdV}, and cold plasma in Figure~\ref{fig:waterfall_cold_KdV}. Compared to the simulations with a traveling wave initial condition, these waves propagate significantly more smoothly. There are no trailing waves except a small offshoot that counterpropagates. And, even though the warbling is still present in the electron velocity, it is appreciably smaller.

From our long-wavelength study, we proved that the two-fluid Euler-Poisson system should converge to solutions of the KdV equations. Moreover, the error-scaled residual will be bounded by a small order of \(\varepsilon\). In Figure~\ref{fig:waterfall_KdV_compare} we compare the evolution of the Korteweg de-Vries solution from \eqref{eq::kdvSol} (blue solid line) to the evolution of the ion density in both hot and cold plasmas (black dotted line). The results validate the theory of KdV being the Euler-Poisson long-wavelength limit; the waves lie on top of one another with minuscule deviations.

\begin{figure}
     \centering
     \begin{subfigure}[b]{0.49\textwidth}
         \centering
         \includegraphics[width=\textwidth, trim={0 40pt 0 70pt}, clip]{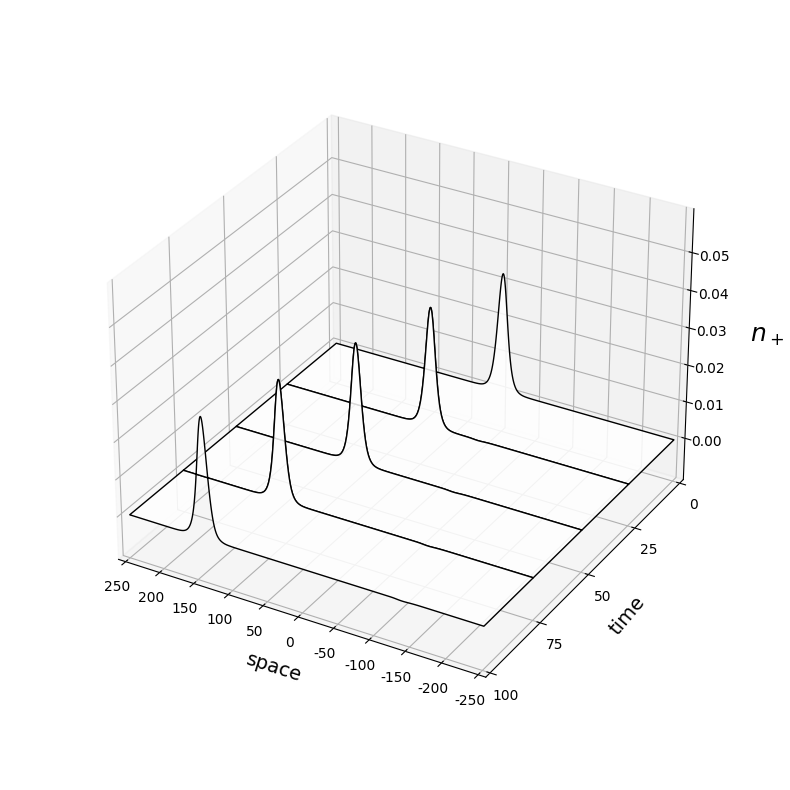}
         \caption{Ion Density, \(n_+(x,t)\) \\[14pt]}
         \label{fig:waterfall_hot_n+_KdV}
     \end{subfigure}
     \hfill
     \begin{subfigure}[b]{0.49\textwidth}
         \centering
         \includegraphics[width=\textwidth, trim={0 40pt 0 70pt}, clip]{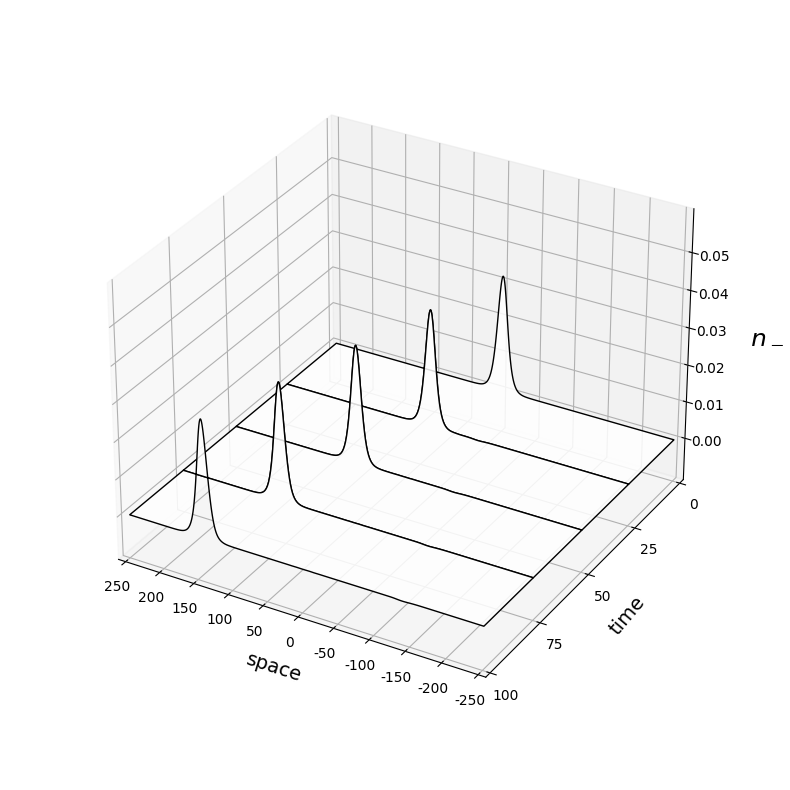}
         \caption{Electron Density, \(n_-(x,t)\) \\[14pt]}
         \label{fig:waterfall_hot_n-_KdV}
     \end{subfigure}
     \hfill
     \begin{subfigure}[b]{0.49\textwidth}
         \centering
         \includegraphics[width=\textwidth, trim={0 40pt 0 70pt}, clip]{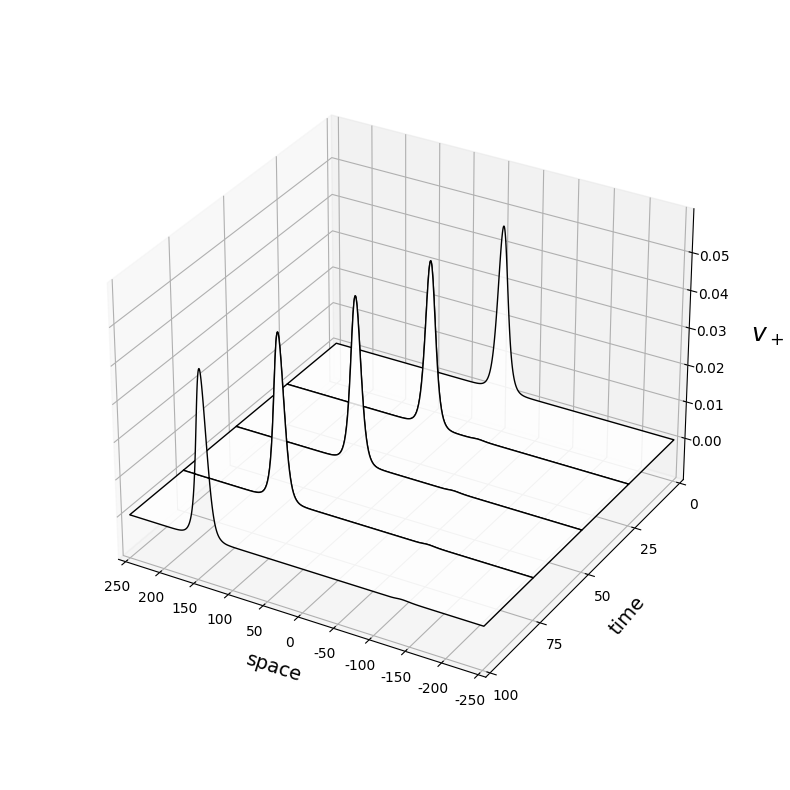}
         \caption{Ion Velocity, \(v_+(x,t)\) \\[14pt]}
         \label{fig:waterfall_hot_v+_KdV}
     \end{subfigure}
     \hfill
     \begin{subfigure}[b]{0.49\textwidth}
         \centering
         \includegraphics[width=\textwidth, trim={0 40pt 0 70pt}, clip]{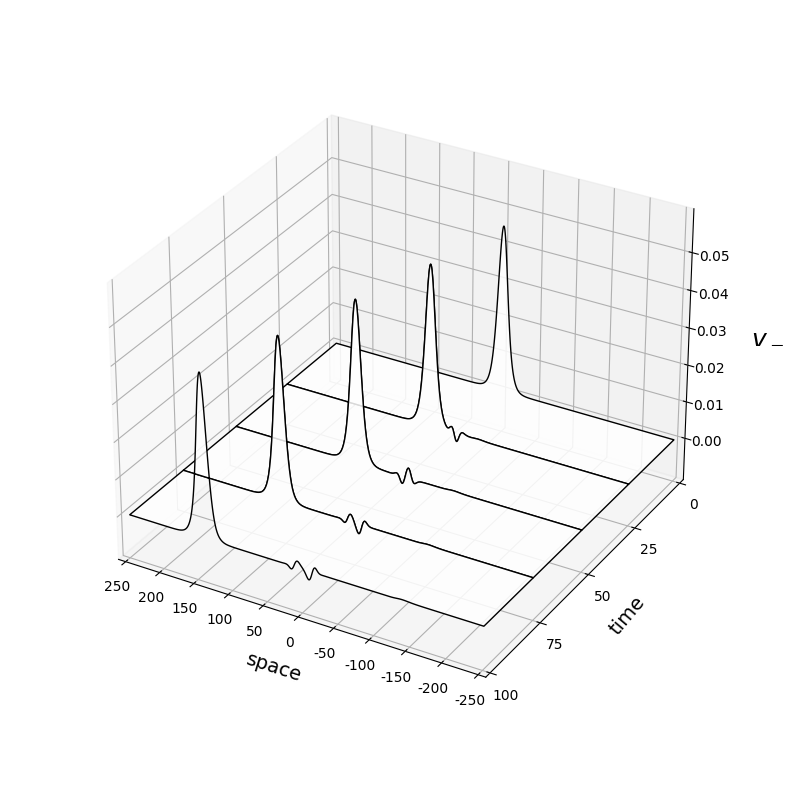}
         \caption{Electron Velocity, \(v_-(x,t)\) \\[14pt]}
         \label{fig:waterfall_hot_v-_KdV}
     \end{subfigure}
     \hfill
     \begin{subfigure}[b]{0.49\textwidth}
         \centering
         \includegraphics[width=\textwidth, trim={0 40pt 0 70pt}, clip]{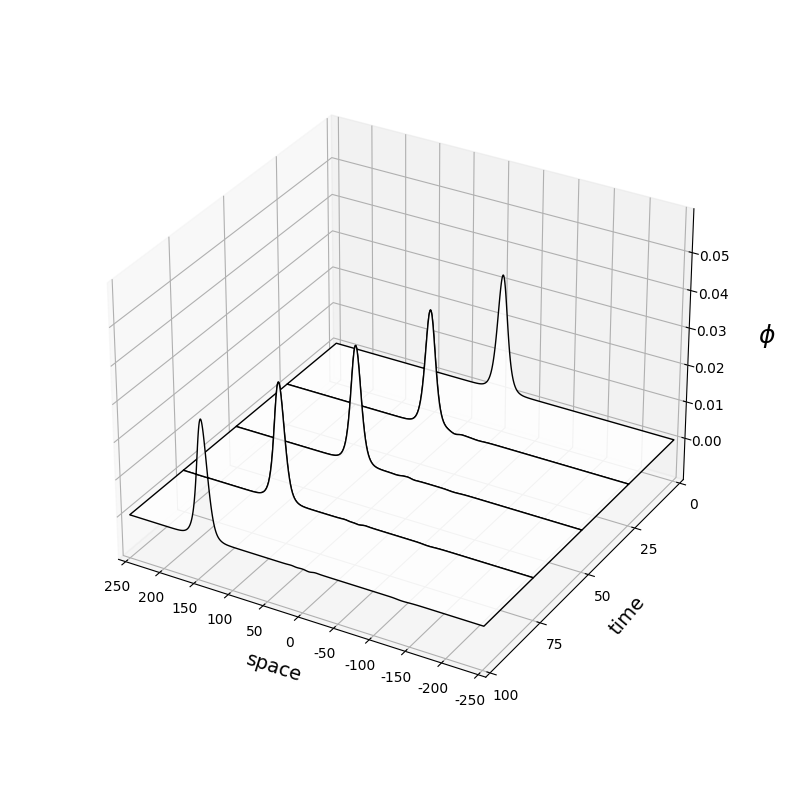}
         \caption{Electric Potential, \(\phi(x,t)\)}
         \label{fig:waterfall_hot_phi_KdV}
     \end{subfigure}
     \caption{Propagation of a hot two-fluid plasma (\(\tau_i=1\)) over space and time. The initial condition is a Korteweg de-Vries wave solution of speed \(\tilde\mu = 1.5\) and \(\varepsilon=10^{-2}\).}
    \label{fig:waterfall_hot_KdV}
\end{figure}

\begin{figure}
     \centering
     \begin{subfigure}[b]{0.49\textwidth}
         \centering
         \includegraphics[width=\textwidth, trim={0 40pt 0 70pt}, clip]{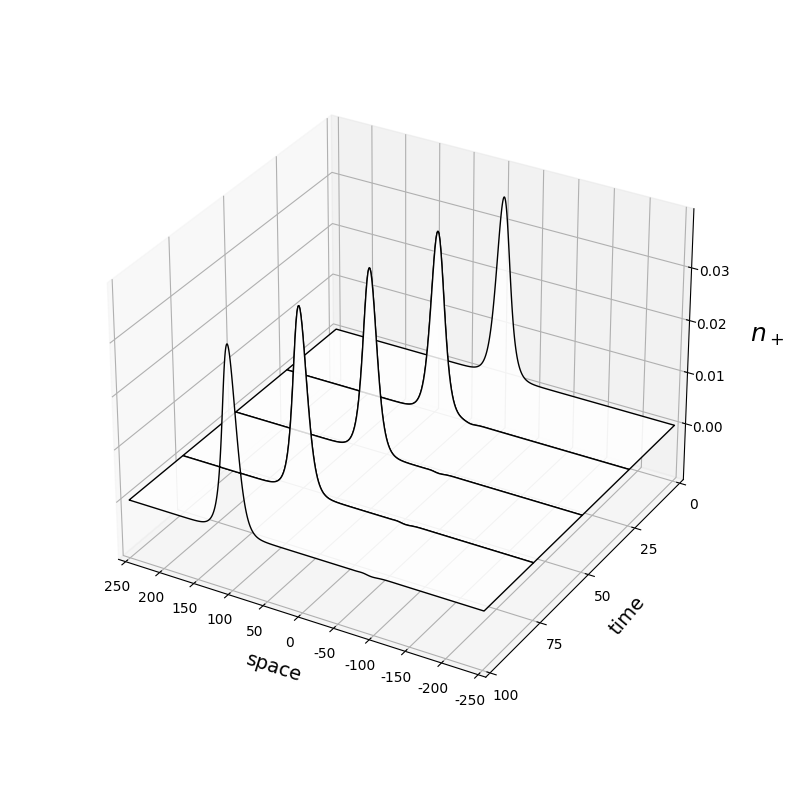}
         \caption{Ion Density, \(n_+(x,t)\) \\[14pt]}
         \label{fig:waterfall_cold_n+_KdV}
     \end{subfigure}
     \hfill
     \begin{subfigure}[b]{0.49\textwidth}
         \centering
         \includegraphics[width=\textwidth, trim={0 40pt 0 70pt}, clip]{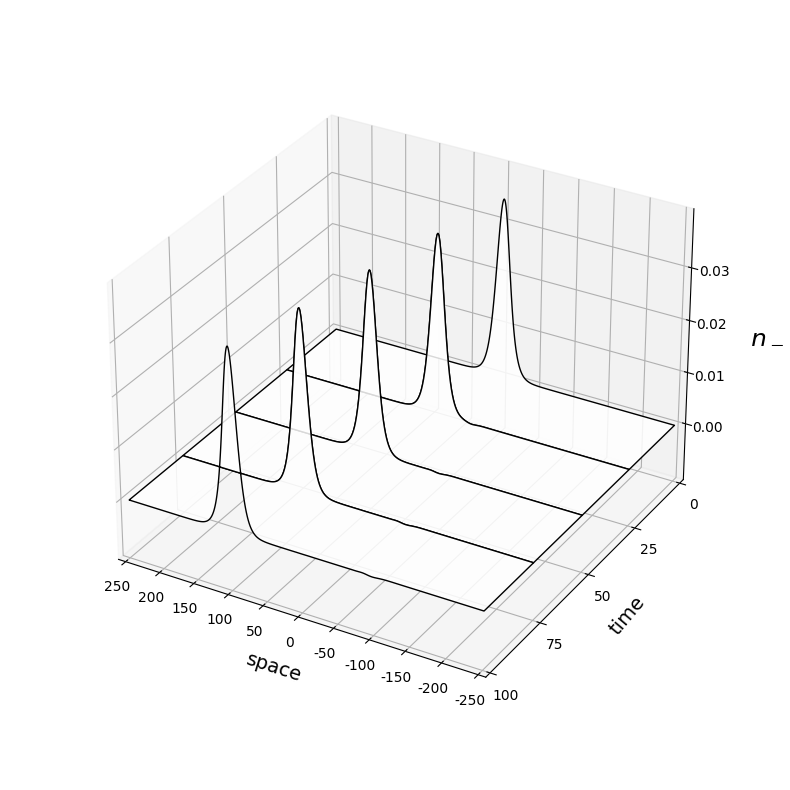}
         \caption{Electron Density, \(n_-(x,t)\) \\[14pt]}
         \label{fig:waterfall_cold_n-_KdV}
     \end{subfigure}
     \hfill
     \begin{subfigure}[b]{0.49\textwidth}
         \centering
         \includegraphics[width=\textwidth, trim={0 40pt 0 70pt}, clip]{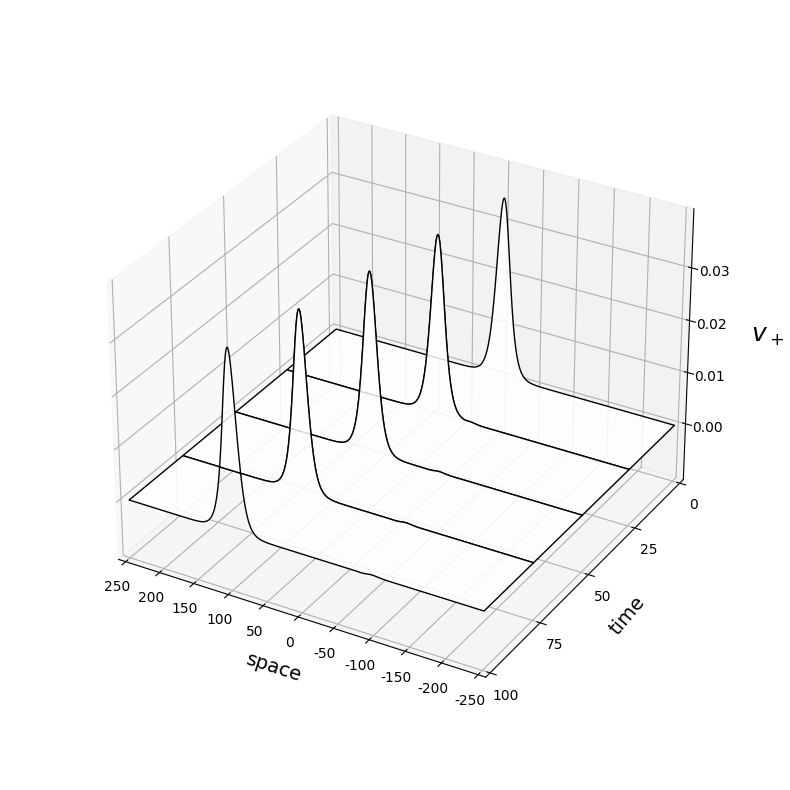}
         \caption{Ion Velocity, \(v_+(x,t)\) \\[14pt]}
         \label{fig:waterfall_cold_v+_KdV}
     \end{subfigure}
     \hfill
     \begin{subfigure}[b]{0.49\textwidth}
         \centering
         \includegraphics[width=\textwidth, trim={0 40pt 0 70pt}, clip]{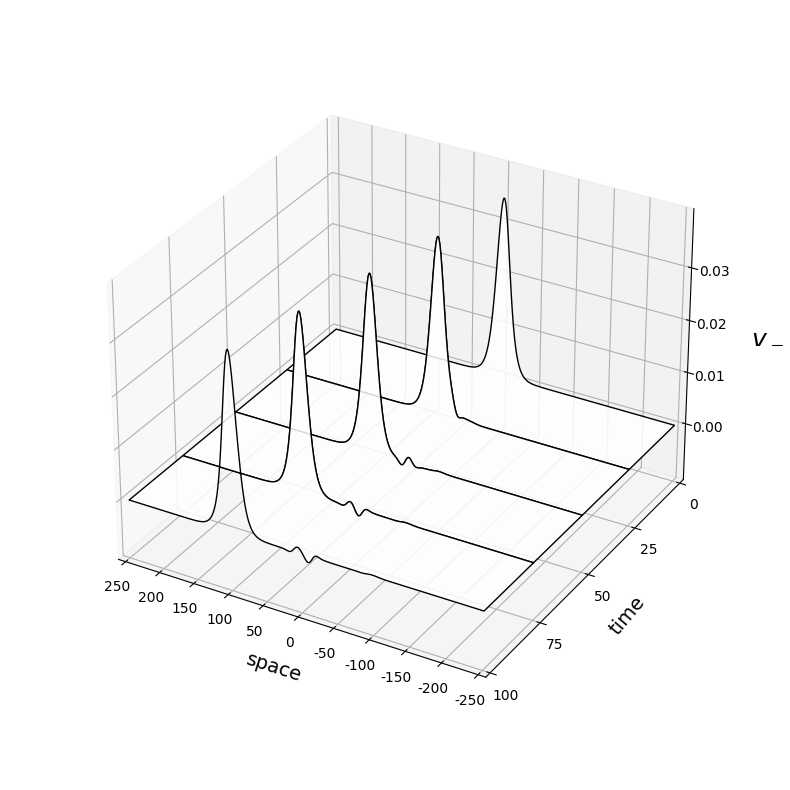}
         \caption{Electron Velocity, \(v_-(x,t)\) \\[14pt]}
         \label{fig:waterfall_cold_v-_KdV}
     \end{subfigure}
     \hfill
     \begin{subfigure}[b]{0.49\textwidth}
         \centering
         \includegraphics[width=\textwidth, trim={0 40pt 0 70pt}, clip]{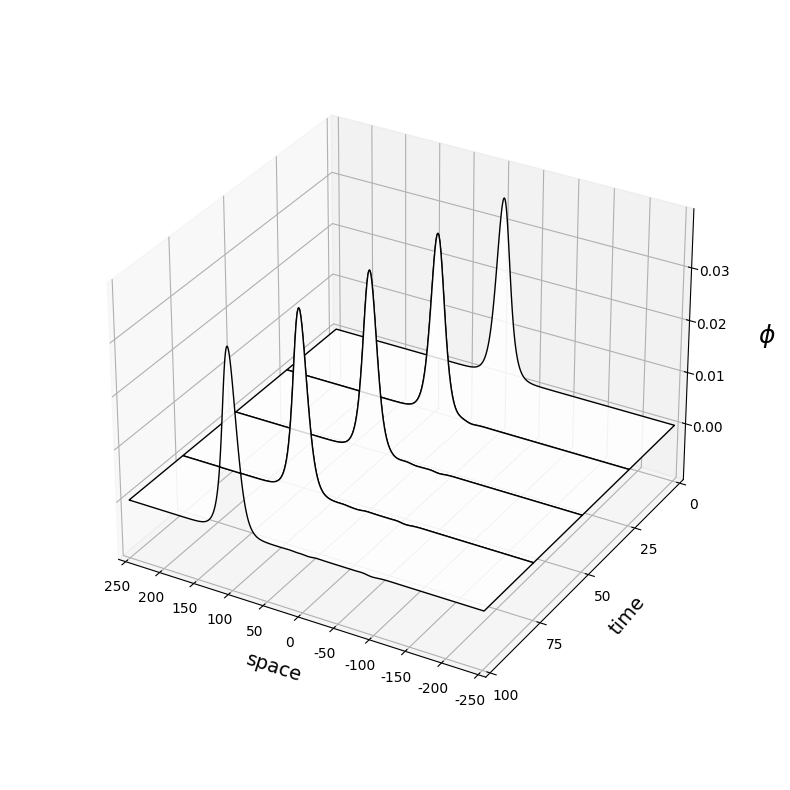}
         \caption{Electric Potential, \(\phi(x,t)\)}
         \label{fig:waterfall_cold_phi_KdV}
     \end{subfigure}
     \caption{Propagation of a cold two-fluid plasma (\(\tau_i=0\)) over space and time. The initial condition is a Korteweg de-Vries wave solution of speed \(\tilde\mu = 1.154\) and \(\varepsilon=10^{-2}\).}
    \label{fig:waterfall_cold_KdV}
\end{figure}

\begin{figure}[h]
\begin{subfigure}[b]{0.49\textwidth}
     \centering
     \includegraphics[width=\textwidth, trim={0 40pt 0 70pt}, clip]{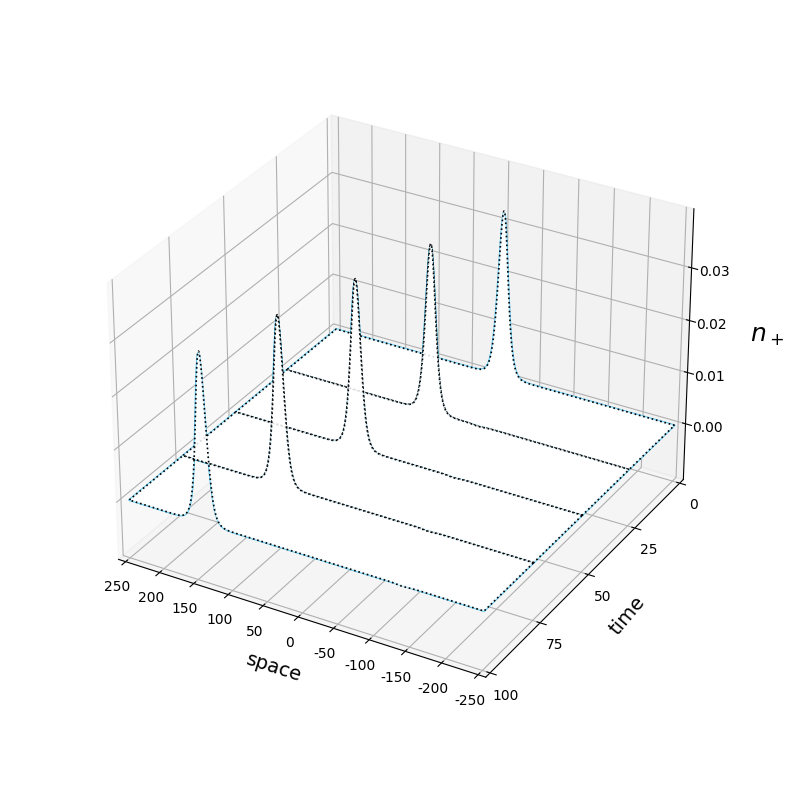}
     \caption{A hot two-fluid plasma (\(\tau_i=1\)).}
         \label{fig:waterfall_hot_compare}
     \end{subfigure}
\hfill
\begin{subfigure}[b]{0.49\textwidth}
     \centering
     \includegraphics[width=\textwidth, trim={0 40pt 0 70pt}, clip]{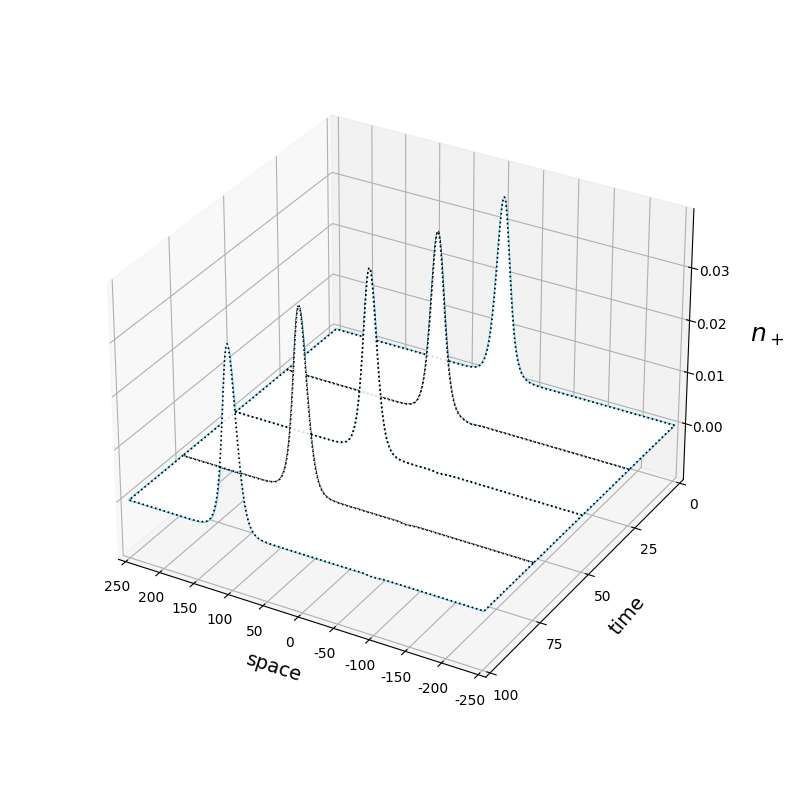}
     \caption{A cold two-fluid plasma (\(\tau_i=0\)).}
         \label{fig:waterfall_cold_compare}
     \end{subfigure}
 \caption{Comparing the evolution of the Korteweg de-Vries solution (blue, solid) to the evolution of \(n_+\) (black, dotted) with \(\varepsilon=10^{-2}\).}
\label{fig:waterfall_KdV_compare}
\end{figure}

\subsection{Numerical Artifacts}

To check if our simulations are consistent, we must observe how the behavior may or may not change with smaller spatial and time steps. We check the electron velocity under various conditions in Figure~\ref{fig:waterfall_other} as this variable exhibits the most extreme phenomena. The initial condition we use is a traveling wave of speed \(\mu = c + 10^{-3}\), a small order higher than sonic speed. This leads to an initial condition of shorter amplitude, and therefore, a tamer warbling in the center of the domain.

The base case with \(\Delta t =0.5\) and \(\Delta x=0.05\) is shown in Figure~\ref{fig:waterfall_smaller_mu}. It bears behavior similar to the other simulations of the same parameters with the higher speed and KdV initial conditions. We start with halving the spatial step, \(\Delta x=0.025\), to get Figure~\ref{fig:waterfall_half_space} and note there is no difference between its propagation and the original. Halving the time step next, \(\Delta t=0.25\), Figure~\ref{fig:waterfall_half_time} shows that while the main wave is a match, the warbling is wider spread. This is also true for halving both the spatial and time steps as shown in Figure~\ref{fig:waterfall_half_both}. This indicates that our simulations are consistent in space, but something fishy might occur in time.

\begin{figure}
     \centering
     \begin{subfigure}[b]{0.49\textwidth}
         \centering
         \includegraphics[width=\textwidth, trim={0 40pt 0 70pt}, clip]{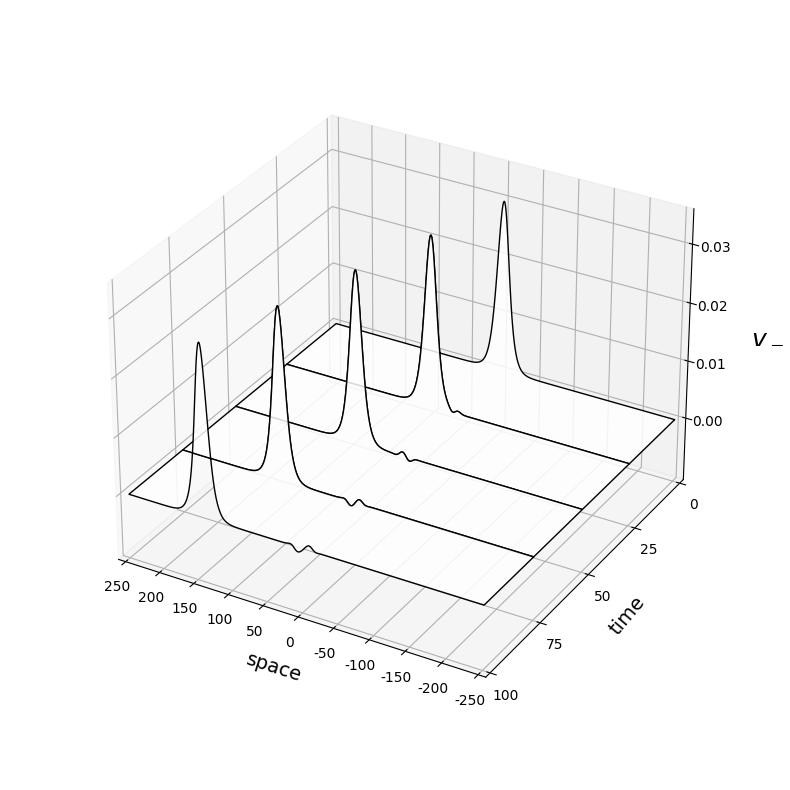}
         \caption{Original, \(\Delta t = 0.5, \, \Delta x = 0.05\) \\[14pt]}
         \label{fig:waterfall_smaller_mu}
     \end{subfigure}
     \hfill
     \begin{subfigure}[b]{0.49\textwidth}
         \centering
         \includegraphics[width=\textwidth, trim={0 40pt 0 70pt}, clip]{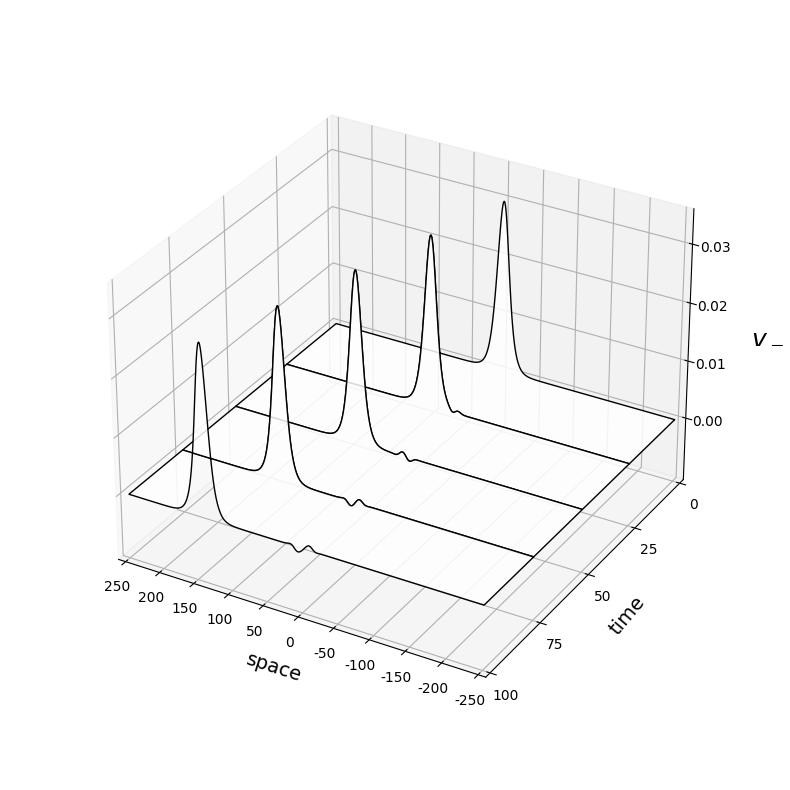}
         \caption{Half spatial step, \( \Delta t = 0.5, \, \Delta x = 0.025\) \\[14pt]}
         \label{fig:waterfall_half_space}
     \end{subfigure}
     \hfill
     \begin{subfigure}[b]{0.49\textwidth}
         \centering
         \includegraphics[width=\textwidth, trim={0 40pt 0 70pt}, clip]{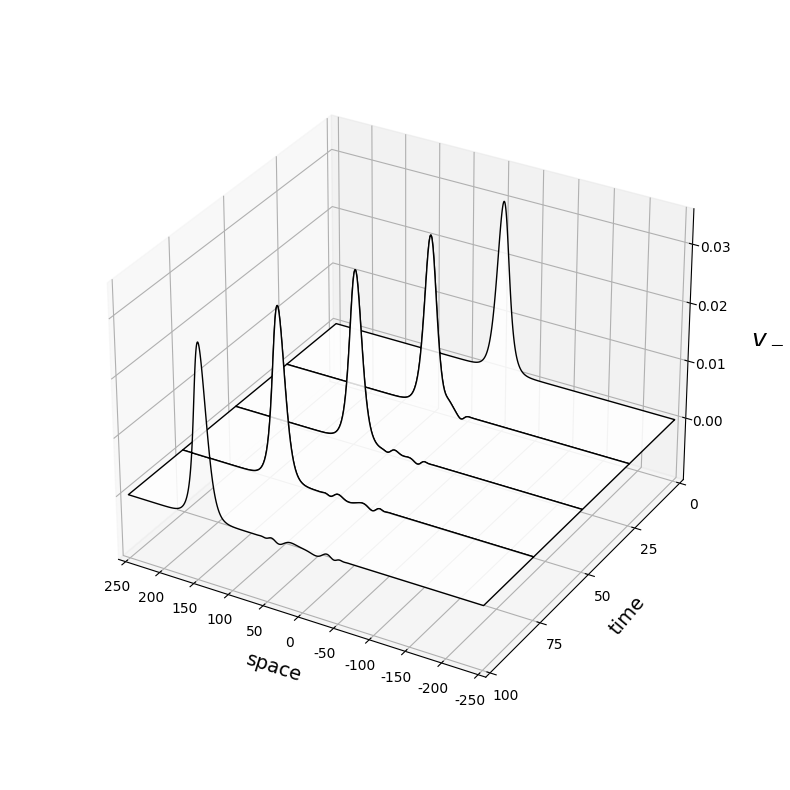}
         \caption{Half time step, \( \Delta t = 0.25, \, \Delta x = 0.05\) \\[14pt]}
         \label{fig:waterfall_half_time}
     \end{subfigure}
     \hfill
     \begin{subfigure}[b]{0.49\textwidth}
         \centering
         \includegraphics[width=\textwidth, trim={0 40pt 0 70pt}, clip]{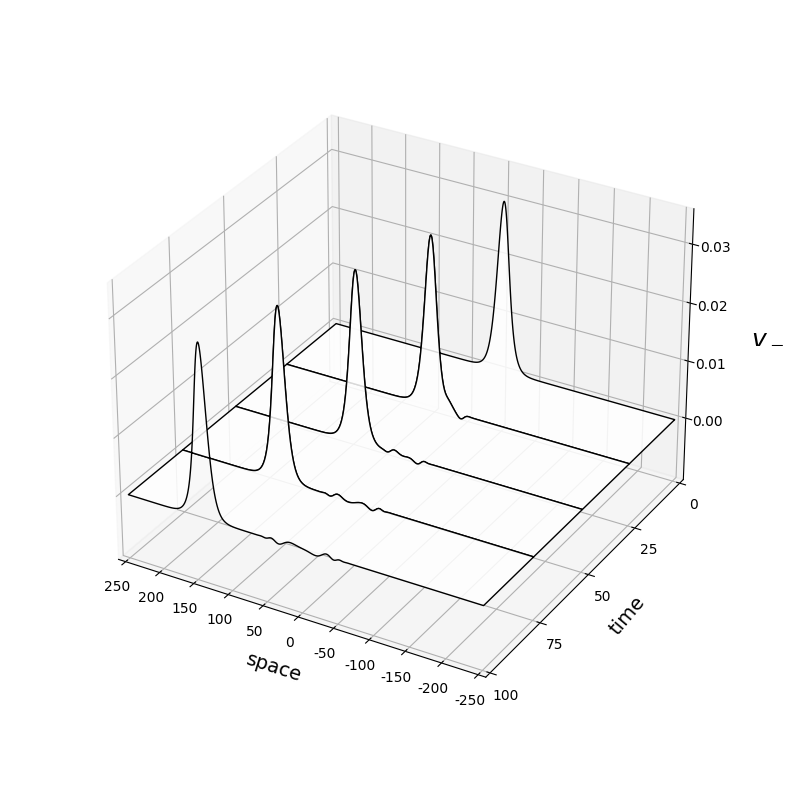}
         \caption{Half both steps, \(\Delta t = 0.25, \, \Delta x = 0.025\) \\[14pt]}
         \label{fig:waterfall_half_both}
     \end{subfigure}
     \caption{Propagation of a hot electron velocity (\(\tau_i=1\)) over space and time with an initial condition of a traveling wave of speed \(\mu = c + 10^{-3}\).}
    \label{fig:waterfall_other}
\end{figure}

Indeed, some of the wild oscillations may be a consequence of the numerical scheme. Because the partial time derivative of the electron velocity is scaled by \(m_e\), it affects the time step we can choose to get the most accurate behavior. Currently, \(m_e < \Delta t \), so the singular perturbation is a larger burden on the system. Thus, we need to take \(\Delta t < m_e\) to limit any resulting numerical artifacts. This process is computationally expensive, however, since \(m_e\) is very small.

We can fight this hurdle by slightly adjusting the electron mass. While it will not stay true to the original scaling of the problem, it gives us a good prototype. We increase the value of \(m_e\) fifty-fold (\(m_e\approx 0.027\)), set \(\Delta t = 0.02\), and run our simulations to a max time of \(T=50\). The results for the electron velocity are given in Figure~\ref{fig:waterfall_big_me}. As suspected, the warbling, although still present, has been subdued to the point that it is barely perceivable.

\begin{figure}[h]
         \centering
         \includegraphics[width=0.75\textwidth, trim={0 40pt 0 70pt}, clip]{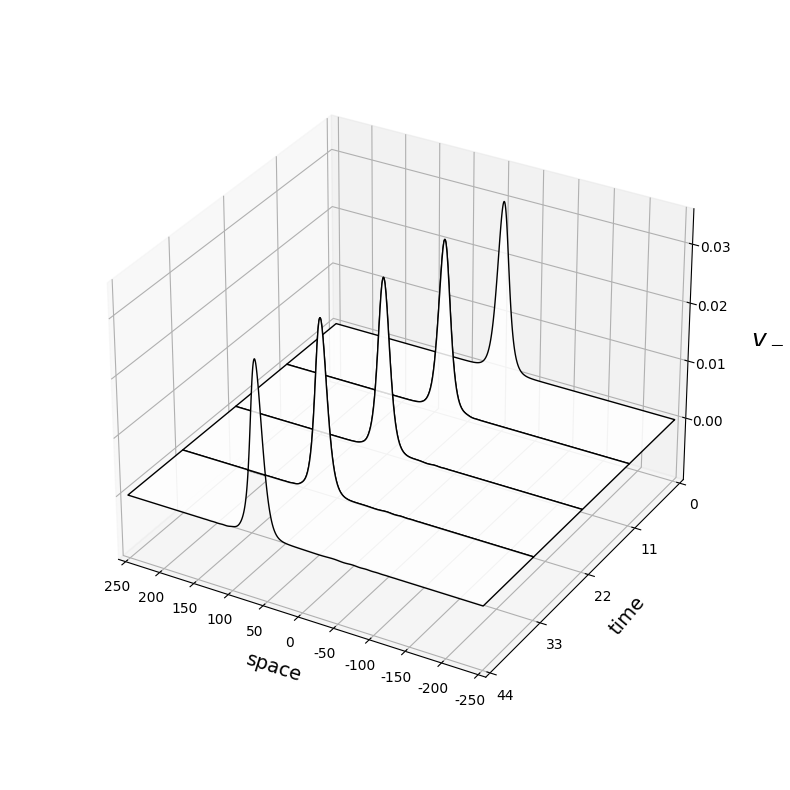}
     \caption{Propogation of hot electron velocity (\(\tau_i=1\)) over space and time with an initial condition of a traveling wave of speed \(\mu = c + 10^{-3}\) with \(\Delta t < m_e\).}
    \label{fig:waterfall_big_me}
\end{figure}

\subsection{Single-Fluid Plasma Simulation}

We will simulate the single-fluid plasma in a similar manner. We begin with \eqref{eq::3hotrescale} and replace \(n_+ \to 1 + n_+\). Then we add and subtract linear terms in the ion velocity equation, regardless of the value of \(\tau_i\), and in the Poisson equation to get
\begin{align*}
\partial_t n_+ & = -\partial_x v_+ -\partial_x(n_+ v_+), \\
\partial_t v_+ & = -\partial_x n_+ - \partial_x\phi - \tfrac{1}{2}\partial_x v_+^2 - \partial_x \ln(1 + n_+) + \partial_x n_+, \\
\phi -\partial_{xx}\phi - n_+ & = -e^{\phi} + \phi + 1 .
\end{align*}
We integrate each equation with respect to time and apply the trapezoidal rule to terms with spatial derivatives. We will also follow the work of Sattinger and Li \cite{li1999soliton}, and average the nonlinear terms in the Poisson equation. This yields
\begin{align*}
    n_{+,2} - n_{+,1} & = \tfrac{D \Delta t}{2}\left( -v_{+,2} - v_{+,1} - (n_{+,2}v_{+,2}) - (n_{+,1}v_{+,1})\right), \\[8pt]
    v_{+,2} - v_{+,1} & = \tfrac{D\Delta t}{2}\big( -\tau_i n_{+,2} -\tau_i n_{+,1} - \phi_2 - \phi_1 - \tfrac{1}{2} v_{+,2}^2 - \tfrac{1}{2} v_{+,1}^2 \\[2pt]
    & \phantom{= \tfrac{D\Delta t}{2}\big( } \; - \tau_i \ln(1 + n_{+,2}) - \tau_i \ln(1 + n_{+,1}) + \tau_i n_{+,2} + \tau_i n_{+,1} \big) ,\\[8pt]
    (1-D^2)\phi_2 - n_{+,2} & = \frac{-e^{\phi_2} + \phi_2 + 1}{2} + \frac{-e^{\phi_1} + \phi_1 + 1}{2}.
\end{align*}
Separating the time-steps and rewriting as an equation of matrices, we obtain
\begin{align*}
    \begin{bmatrix}
    1 & \tfrac{1}{2}D\Delta t & 0 \\[4pt]
    \tfrac{1}{2}D\Delta t & 1 & \tfrac{1}{2}D\Delta t \\[4pt]
    -1 & 0 & 1-D^2
    \end{bmatrix}
    \begin{bmatrix}
    n_{+,2} \\[4pt]
    v_{+,2} \\[4pt]
    \phi_2
    \end{bmatrix}
    &
    =
    \begin{bmatrix}
    1 & -\tfrac{1}{2}D\Delta t & 0 \\[4pt]
    -\tfrac{\tau_i }{2}D\Delta t & 1 & -\tfrac{1}{2}D\Delta t \\[4pt]
    0 & 0 & 0
    \end{bmatrix}
    \begin{bmatrix}
    n_{+,1} \\[4pt]
    v_{+,1} \\[4pt]
    \phi_1
    \end{bmatrix}
    \\[10pt]
    & \quad
    +
    \begin{bmatrix}
    \tfrac{1}{2}D\Delta t(- n_{+,2}v_{+,2}) \\[8pt]
    \tfrac{1}{2}D\Delta t\left(- \tfrac{1}{2}v_{+,2}^2 - \tau_i \ln(1+n_{+,2}) + \tau_i n_{+,2}\right) \\[8pt]
    \tfrac{1}{2}(-e^{\phi_2} + \phi_2 + 1)
    \end{bmatrix} \\[10pt]
    & \quad
    +
    \begin{bmatrix}
    \tfrac{1}{2}D\Delta t(- n_{+,1}v_{+,1}) \\[8pt]
    \tfrac{1}{2}D\Delta t\left(- \tfrac{1}{2}v_{+,1}^2 - \tau_i \ln(1+n_{+,1}) + \tau_i n_{+,1}\right) \\[8pt]
    \tfrac{1}{2}(-e^{\phi_1} + \phi_1 + 1)
    \end{bmatrix}.
\end{align*}
Utilizing the same scheme as \eqref{eq::future_time_calc} and similar code, we can solve this matrix equation to observe the single-fluid system's behavior in hot and cold plasmas. Our observations follow.

\subsubsection{Ion Density \& Velocity and Electric Potential}

The ion density, velocity, and electric potential in a hot and cold single-fluid plasma are shown in Figures~\ref{fig:waterfall_me0_hot} and \ref{fig:waterfall_me0_cold}, respectively. While they bear similar overall behaviors to the two-fluid system, notable distinctions emerge upon closer inspection. First, the trailing wave behavior is not as pronounced. The counter-propagating waves also decrease in magnitude, are less intense, and do not move as fast as their two-fluid counterparts. Moreover, the warbling in the middle of the electric potential term disappears. This loss is due to eradicating the singular perturbation by \(m_e\) in the Euler-Poisson system. 

\begin{figure}[h!]
     \centering
     \begin{subfigure}[b]{0.49\textwidth}
         \centering
         \includegraphics[width=\textwidth, trim={0 25pt 0 50pt}, clip]{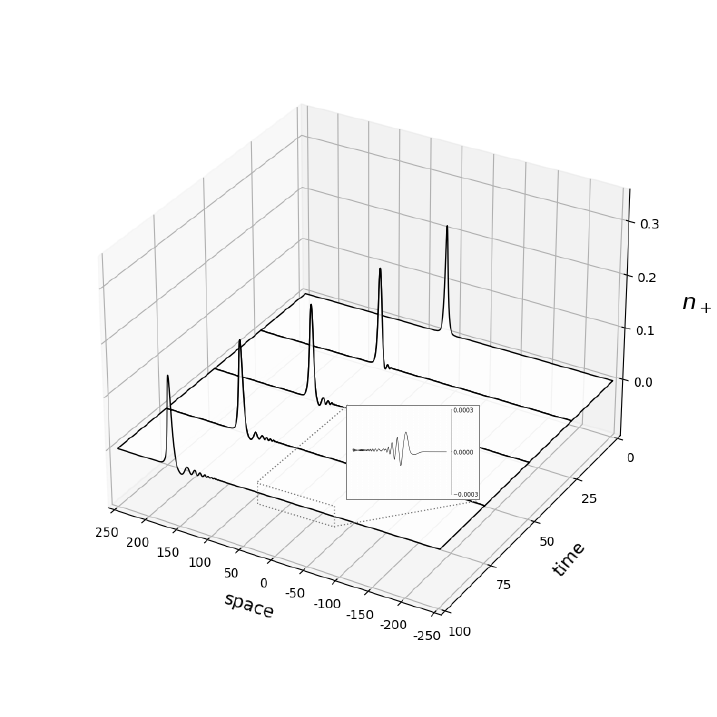}
         \caption{Ion Density, \(n_+(x,t)\) \\[14pt]}
         \label{fig:waterfall_me0_hot_n+}
     \end{subfigure}
     \hfill
     \begin{subfigure}[b]{0.49\textwidth}
         \centering
         \includegraphics[width=\textwidth, trim={0 40pt 0 70pt}, clip]{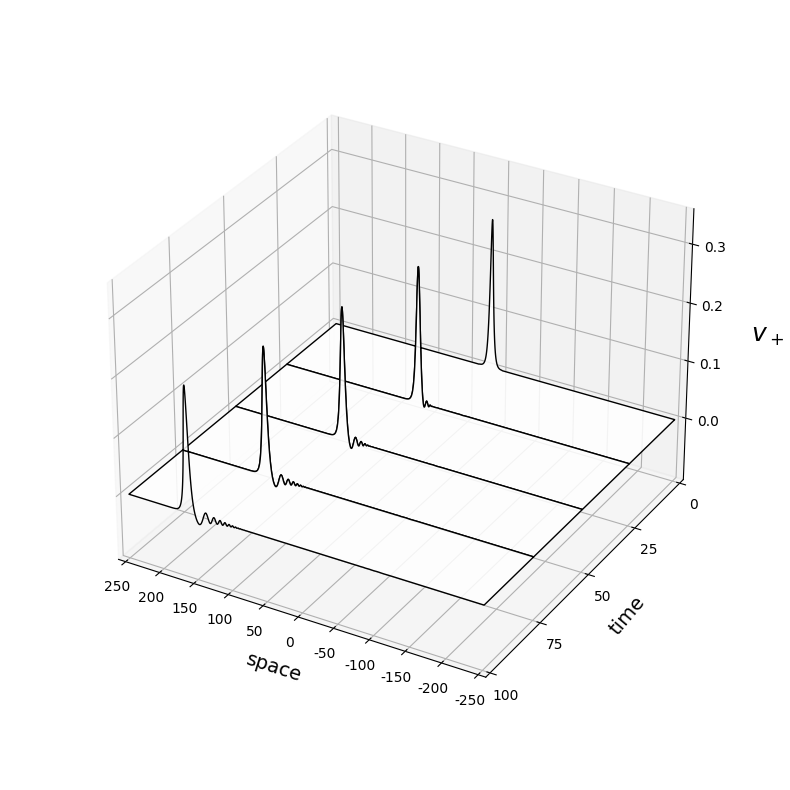}
         \caption{Ion Velocity, \(v_+(x,t)\) \\[14pt]}
         \label{fig:waterfall_me0_hot_v+}
     \end{subfigure}
     \hfill
     \begin{subfigure}[b]{0.49\textwidth}
         \centering
         \includegraphics[width=\textwidth, trim={0 40pt 0 70pt}, clip]{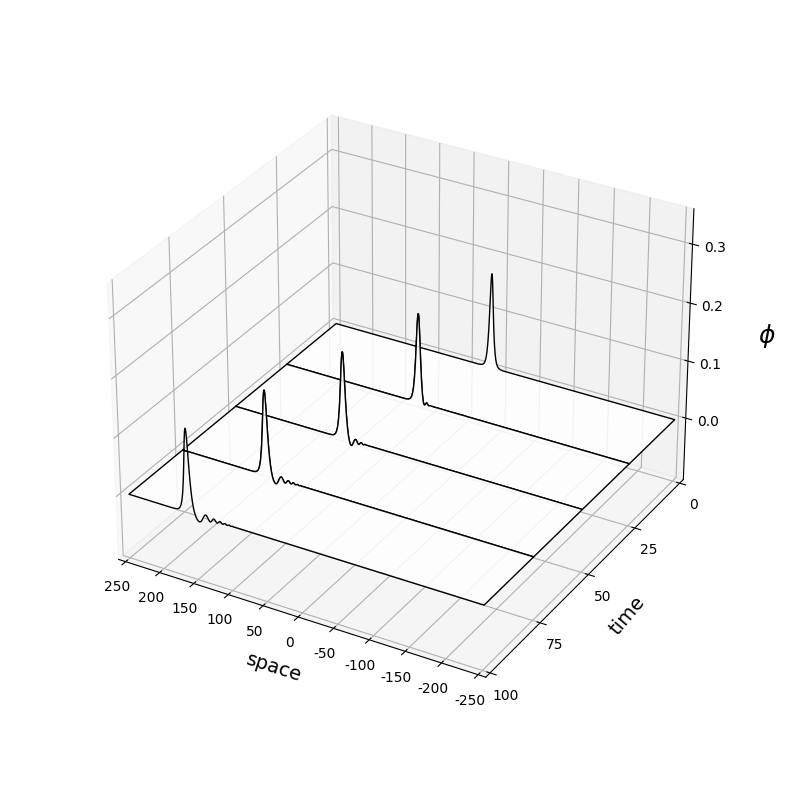}
         \caption{Electric Potential, \(\phi(x,t)\)}
         \label{fig:waterfall_me0_hot_phi}
     \end{subfigure}
     \caption{Propagation of a hot single-fluid plasma (\(\tau_i=1, \, m_e =0\)) over space and time. The initial condition is a traveling wave solution of speed \(\mu = 1.5\).}
    \label{fig:waterfall_me0_hot}
\end{figure}
\begin{figure}[h!]
     \centering
     \begin{subfigure}[b]{0.49\textwidth}
         \centering
         \includegraphics[width=\textwidth, trim={0 25pt 0 50pt}, clip]{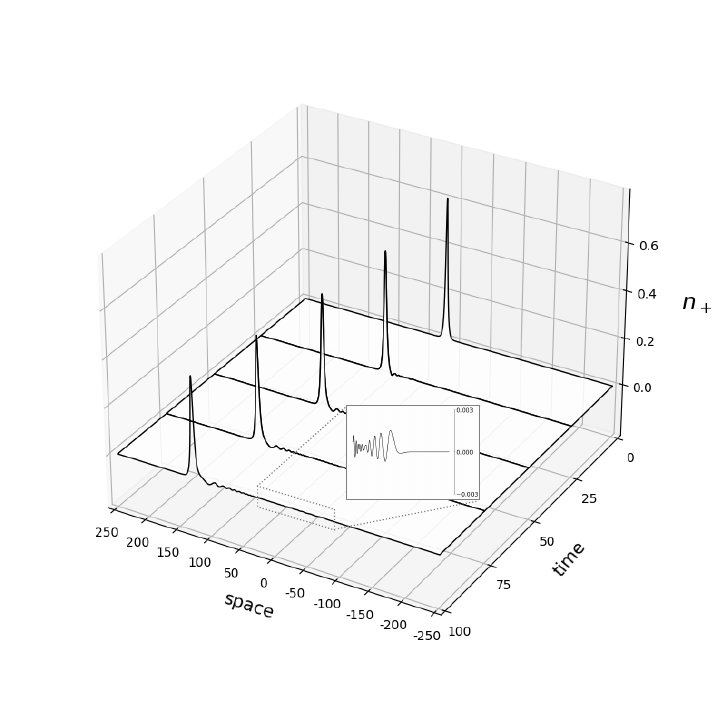}
         \caption{Ion Density, \(n_+(x,t)\) \\[14pt]}
         \label{fig:waterfall_me0_cold_n+}
     \end{subfigure}
     \hfill
     \begin{subfigure}[b]{0.49\textwidth}
         \centering
         \includegraphics[width=\textwidth, trim={0 40pt 0 70pt}, clip]{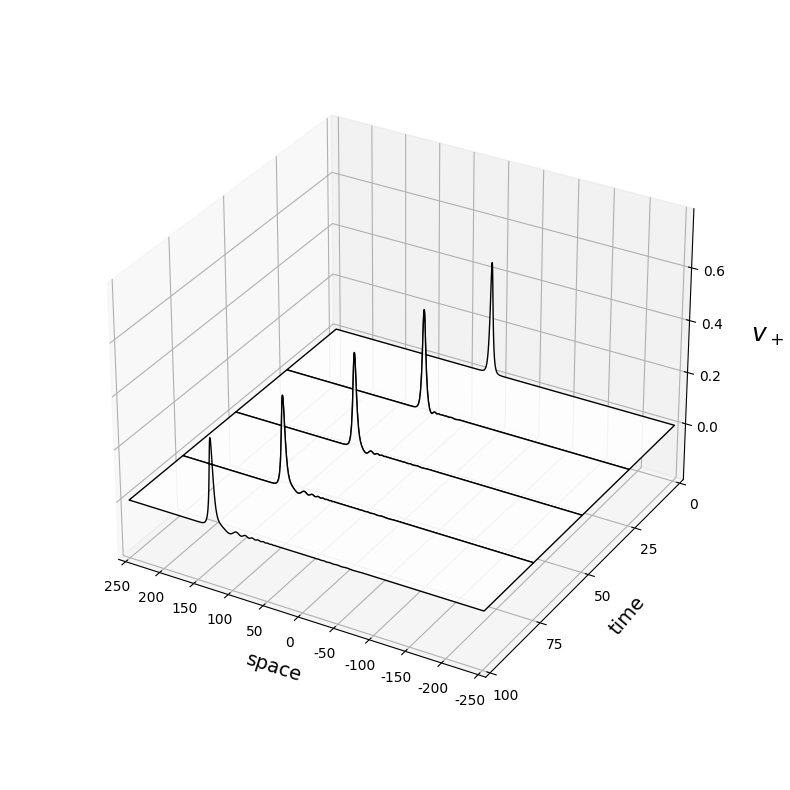}
         \caption{Ion Velocity, \(v_+(x,t)\) \\[14pt]}
         \label{fig:waterfall_me0_cold_v+}
     \end{subfigure}
     \hfill
     \begin{subfigure}[b]{0.49\textwidth}
         \centering
         \includegraphics[width=\textwidth, trim={0 40pt 0 70pt}, clip]{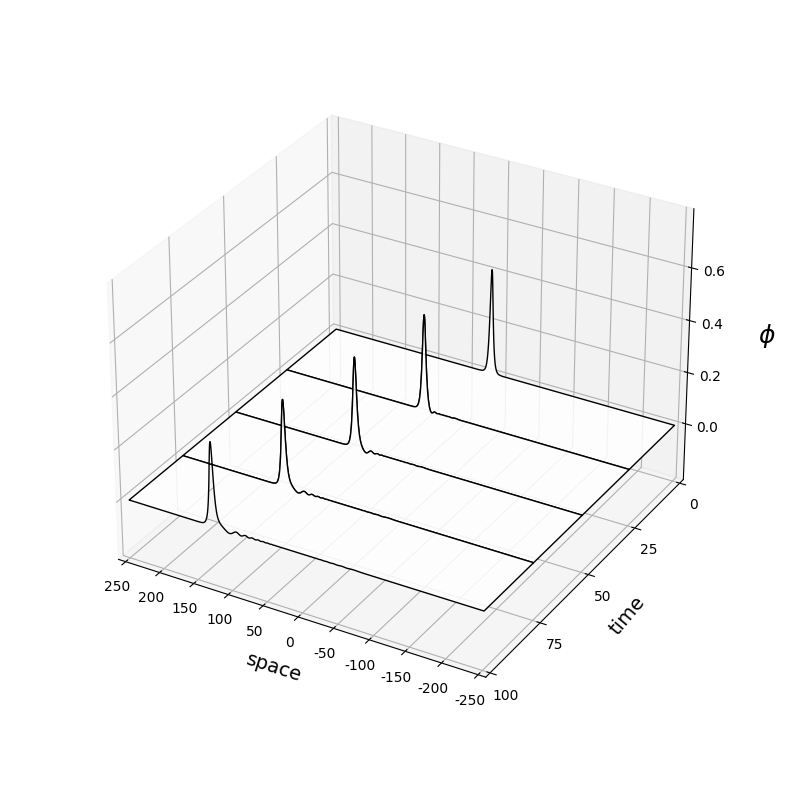}
         \caption{Electric Potential, \(\phi(x,t)\)}
         \label{fig:waterfall_me0_cold_phi}
     \end{subfigure}
     \caption{Propagation of a cold single-fluid plasma (\(\tau_i=0, \, m_e =0\)) over space and time. The initial condition is a traveling wave solution of speed \(\mu = 1.154\).}
    \label{fig:waterfall_me0_cold}
\end{figure}
\section{Conclusion and Future Directions}\label{conclusion}

Our research studied the dynamics of ion-acoustic waves in a two-fluid plasma, a domain relatively unexplored compared to past studies focusing solely on single-fluid plasmas. By considering the mass discrepancy between electrons and ions, we sought to unravel how this fundamental difference impacts wave behavior within the medium. Furthermore, our investigation explored the contrasting characteristics of hot and cold plasmas, manipulating ion temperatures to observe their distinct effects on wave dynamics. Through these comprehensive analyses, we aimed to contribute to a deeper understanding of plasma physics.

Indeed, we demonstrated how the two-fluid Euler-Poisson system has traveling wave solutions for specific wave speeds. Leveraging solutions obtained through traveling wave and KdV calculations, we established initial conditions for our subsequent numerical simulations. We then visualized the propagation of ion-acoustic waves and other dynamic phenomena by employing an algorithm utilizing Fast Fourier Transforms. Through this process, we observed the distinct contributions of the key parameters, \(m_e\) and \(\tau_i\), on plasma behavior.

Among the noteworthy discoveries, we observed striking parallels between the two-fluid system when \(m_e\) is set to zero and the theoretical outcomes derived from the single-fluid model. However, our investigations uncovered a significant caveat: nullifying the electron mass obscures a comprehensive understanding of plasma behavior. A particularly intriguing finding was the manifestation of the singular perturbations in the velocity of the electrons, resulting in erratic oscillations that profoundly influenced the electric potential in the system. 

Furthermore, our findings transcend the constraints of specific ion temperatures, broadening the applications of our results to other plasmas characterized by temperature. While studying the difference between hot and cold plasmas, we noted how the behaviors of the waves are affected. Even though the overall trends in wave behavior remain qualitatively similar across different thermal environments, quantitative deviations in amplitude, propagation, and velocity highlight the medium's sensitivity to temperature variations.

Looking ahead in our research journey, our focus shifts toward refining the robustness of our findings. In our exploration of traveling waves, we relied on an estimated electron mass for calculating a key value of an inverse function value. Ideally, we aim to circumvent this approach, aligning with the core objective of our project. Ultimately, we are looking to extend this study to higher dimensions to unveil more insights into the Euler-Poisson system's underlying dynamics. By venturing into higher-dimensional spaces, we anticipate uncovering novel patterns, phenomena, and emergent behaviors that may have eluded detection within the confines of lower-dimensional frameworks.

\newpage
\bibliographystyle{siam}

\end{document}